\documentclass[11pt]{article}

\newcommand{\mypara}[1]{\smallskip\noindent\textbf{#1.}}  % Paragraph headers
\newcommand{\subsec}[1]{\subsection{#1}}

\usepackage[text={6.5in,9in},centering]{geometry}
\usepackage[cmex10]{amsmath}
\usepackage{enumerate}
\usepackage{graphicx}
\usepackage{amsfonts}
\usepackage{amsthm}
\usepackage{url}
\usepackage{gensymb}
\usepackage{pdfpages}

\newtheorem{proposition}{Proposition}
\newtheorem{lemma}{Lemma}
\newtheorem{corollary}{Corollary}
\newtheorem{theorem}{Theorem}
\newtheorem{definition}{Definition}

\newtheorem{claim}{Claim}
\newtheorem{fact}{Fact}
\newtheorem{axiom}{Axiom}

\newcommand{\capacity}{\textsf{Capacity}}
\newcommand{\wcapacity}{\textsf{WCapacity}}
\newcommand{\scheduling}{\textsf{Scheduling}}
\newcommand{\tlog}{\widehat{\log}}

\newcommand{\cG}{{\cal G}}
\newcommand{\cH}{{\cal H}}
\newcommand{\cK}{{\cal K}}

\begin{document}

\begin{titlepage}

\title{How Well Can Graphs Represent Wireless Interference?}

\author{
  Magn\'us M. Halld\'orsson
  \qquad
  Tigran Tonoyan \\ \\
  ICE-TCS, School of Computer Science \\
  Reykjavik University \\
  \url{mmh@ru.is, ttonoyan@gmail.com}
}

\maketitle
\thispagestyle{empty}

\begin{abstract}
  Efficient use of a wireless network requires that transmissions be grouped into feasible sets, where feasibility means
  that each transmission can be successfully decoded in spite of the interference caused by simultaneous transmissions.
  Feasibility is most closely modeled by a signal-to-interference-plus-noise (SINR) formula, which unfortunately is
  conceptually complicated, being an asymmetric, cumulative, many-to-one relationship.

We re-examine how well graphs can capture wireless receptions as encoded in SINR relationships, placing them in a
framework in order to understand the limits of such modelling.  We seek for each wireless instance a pair of graphs that
provide upper and lower bounds on the feasibility relation, while aiming to minimize the gap between the two graphs.
The cost of a graph formulation is the worst gap over all instances, and the price of (graph) abstraction is the
smallest cost of a graph formulation.

We propose a family of conflict graphs that is parameterized by a non-decreasing sub-linear function, and show that with
a judicious choice of functions, the graphs can capture feasibility with a cost of $O(\log^* \Delta)$, where $\Delta$ is
the ratio between the longest and the shortest link length.  This holds on the plane and more generally in doubling
metrics.  We use this to give greatly improved $O(\log^* \Delta)$-approximation for fundamental link scheduling
problems with arbitrary power control.

We explore the limits of graph representations and find that our upper bound is tight: the price of graph abstraction is
  $\Omega(\log^* \Delta)$.  We also give strong impossibility results for general metrics, and for approximations in
  terms of the number of links.
\end{abstract}

\end{titlepage}

\section{Introduction}

% Wireless interference, scheduling
\mypara{Wireless scheduling} At the heart of any wireless network is a mechanism for managing \emph{interference}
between simultaneous transmissions.  The medium access (MAC) layer manages access to the shared resource, the wireless
spectrum, balancing the aim of maximizing simultaneous use with the impact of the resulting interference.  We can
represent a transmission as a communication \emph{link}, a sender-receiver pair of nodes in a metric space.  Wireless
scheduling mechanisms assign the links to different ``slots'', involving different frequencies, phases and/or time steps.

%  SINR model, feasibility
The model of communication that most closely captures actual conditions, nicknamed the \emph{physical model}, uses a
formula based on the ratio of the (intended) signal strength to the received interference strength (SINR) to determine
if decoding is successful.  A subset $X$ of links is \emph{feasible} if there exists a power assignment to the senders
such that each link $i$ satisfies the SINR formula $\frac{\sum_{j \in X} I_{ji} + N}{S_i} \le \beta^{-1}$ within its
subset, where $S_i$ is the received signal strength on link $i$, $\beta$ and $N$ are fixed constants (dependent on
technology and environment), and $I_{ji}$ is the interference strength of link $j$ on link $i$ \cite{kumar00} (see
Sec.~\ref{S:sinr} for full definitions). We can avoid dealing directly with power assignments using a condition for
feasibility due to Kesselheim \cite{kesselheimconstantfactor} (and shown here to be necessary).  The point remains, however,
that the feasibility predicate is an asymmetric, cumulative, many-to-one relationship.

% Modeling with graphs
\mypara{Capturing interference with graphs}
The aim of this work is to capture the complex feasibility relationship of the physical model with \emph{graphs}, a much
more amenable and better studied model. Preferably, a feasible set of links should correspond to an \emph{independent
  set} in the graph on the links, and vice versa, but since exact capture is impossible, we seek instead approximate
representations. Specifically, we want a formulation that constructs a \emph{pair} of graphs that bound feasibility both
from below and from above. The two graphs should also be ``close'' in some sense; specifically we want an independent
set in the lower bound graph to induce a low chromaticity subgraph in the upper bound graph.  That is, given a set $L$
of links, form graphs $G_l(L)$ and $G_u(L)$ on $L$ such that:
% ... the formulation specifies graphs ...
\begin{itemize} 
\setlength{\itemsep}{0cm}%
  \setlength{\parskip}{0cm}%
 \item If $S$ is feasible subset of links, then it is an independent set in $G_l(L)$,
 \item If $S$ is an independent set in $G_u(L)$, then $S$ is feasible set of links, and
 \item Chromatic numbers of the subgraphs induced by a subset $S$ are close, or $\max_{S \subseteq L} \frac{\chi(G_u[S])}{\chi(G_l[S])} \le \rho$.
\end{itemize}
We may dub the worst ratio $\rho$ over all instances as the \emph{cost} of that graph formulation, the price we pay for using that simpler pairwise and binary graph representation. The least such cost over all graph formulation can then be called the \emph{price of (graph) abstraction}.

% Our results: Conflict graph formulations; upper and lower bounds.
\mypara{Our results}
We propose a family of conflict graph representations, parameterized by a sub-linear, non-decreasing function $f$.
It generalizes known families, e.g., disc graphs correspond to linear functions $f$ and pairwise feasibility is captured by constant functions. The graphs in the family have a structural property that allows for effective approximability.
We argue that this family captures all meaningful conflict graph representations, modulo constant factors. Our main positive result is that for the right choice of $f$, our conflict graph representation has a cost of $O(\log^* \Delta)$, where $\Delta$ is the ratio between longest and the shortest link length. 
We also show that all meaningful representations must pay this $\log^* \Delta$ factor.
Thus, the price of abstraction, for the SINR model with arbitrary power control, is $\Theta(\log^* \Delta)$.
The upper bounds hold for planar instances, and more generally in doubling metrics, while the lower bounds are on the line. We also find that no such results are possible in general metrics nor in terms of the parameter $n$, the number of links.

% Applications: Scheduling and WCapacity
We apply our formulations to obtain greatly improved bounds for fundamental wireless scheduling problems.
In the link {\scheduling} problem, we want to partition a given set of links into fewest possible feasible
sets. In the \emph{weighted capacity} problem, {\wcapacity}, the links have associated positive weight, and 
we want to find the maximum weight feasible subset. In both cases, our $O(\log^* \Delta)$-approximations 
are the first sub-logarithmic approximations known. 

\mypara{Related work}
% Models
Gupta and Kumar \cite{kumar00} proposed the geometric version of the SINR model, where signal decays as a fixed
polynomial of distance; it has since been the default in analytic and simulations studies.
They also initiated the average-case analysis of network capacity, giving rise to a large body on ``scaling laws''.
Moscibroda and Wattenhofer \cite{moscibrodaconnectivity} initiated worst-case analysis in the SINR model. 

% Graph-based models
Graph-based models of wireless communication have been very common in the past.
Most common are geometric graphs involving circular ranges: \emph{unit-disc graphs} (UDG) have all ranges of the same
radius, while in \emph{disc graphs}, the radius can vary with the power assigned, and in the \emph{protocol} model
\cite{kumar00} the communication and interference ranges are different.  Various attempts have been made to add realism,
such as with 2-hop model, or \emph{quasi-unit disc graphs} \cite{barriere2003robust,kuhn2003ad}, and the recently
studied model of \emph{dual graphs} \cite{kuhn2010broadcasting} captures arbitrary unreliability in networks.  None of
these known models offers though any guarantees on fidelity for representing SINR relationships, see e.g.\
\cite{moscibeyondgraphs}.  Graph formulations for modeling SINR relationships were given previously in \cite{us:talg12}
followed by \cite{tonoyanmeancapacity}, but the cost factor was either $O(\log\log \Delta \cdot \log n)$ or $O(\log
\Delta)$.

% Scheduling problem: early work + NP-hardness
Early work on the {\scheduling} problem includes \cite{CruzS03,elbatt02,chafekar07,goussevskaiacomplexity}.
NP-completeness results have been given for different variants \cite{goussevskaiacomplexity, katz2010energy,lin2012complexity}, but as of yet no APX-hardness or stronger lower bounds are known for any related problem in geometric settings, perhaps indicating the difficulty of dealing with the SINR constraints.
%
% Capacity and 
The related {\capacity} problem, where we seek to find a maximum feasible subset of links, admits constant-factor
approximation \cite{kesselheimconstantfactor}.  This immediately implies a $O(\log n)$-approximation for {\scheduling},
where $n$ is the number of links.  Another approach is to solve links of similar lengths in groups, which results in a
$O(\log \Delta)$-approximation \cite{goussevskaiacomplexity,fu2009power,us:talg12}.  The question of improved
approximation for {\scheduling} has been frequently cited as perhaps the most fundamental open problem in the field
\cite{locher2008sensor,dagstuhl11091,fanghanel2011scheduling,GoussevskaiaHW14,halldorsson2014making}. 
The {\wcapacity} problem has applications in several extensions of unit-demand link scheduling problems, such as \emph{stochastic packet scheduling}~\cite{TE92,ryu2013dss,kesselheimStability}, \emph{general demand vectors, multi-path flow} etc.~\cite{WanFJYXT11}. A more general variant of this problem has been considered in the framework of \emph{combinatorial auctions}~\cite{HoeferKV11, HoeferK12}.
{\scheduling} and {\wcapacity} have also been considered with \emph{fixed oblivious power assignments}~\cite{halwatapx, fangkesoblivious, us:talg12,
  halmitcognitive,fangkeslinear}, but the only known sub-logarithmic approximations are known in the case of the
\emph{linear power scheme}~\cite{halmitcognitive, tonoyanlinear}.

\mypara{Issues of models}
% Simpler representations in different fields
A wide range of areas in computer science and mathematics deal with finding simpler abstractions of complex
phenomena. Some examples include: a) discrepancy theory, b) dimensionality reductions and embeddings, c) graph
augmentations and sandwiching properties, d) graph sparsification, e) curve fitting (including least squares, finite
methods, and regression), f) approximation theory (in math), including generalized Fourier series and Chebyshev
polynomials, and g) PAC learning.

% Coarse vs. brittle models
There are tradeoffs between the accuracy of a model and its complexity of detail.
There are legitimate concerns that models and problem formulations are sometimes overly detailed and ``brittle'',
possibly exceeding reasonable levels of precision (see, e.g., discussion in \cite{papadimitriou1997np}).
The benefits of a coarser model tend to include simpler algorithms, easier analysis, but also less sensitivity to incidental details that may or may not be modelable.

% SINR model issues
A case in point is the SINR model, which has its issues.  Whereas the additivity of interference and the near-threshold
nature of signal reception has been borne out in experiments, the geometric decay assumption is far off in essentially
all actual environments \cite{son2006,MaheshwariJD08,sevani2012sir,us:MSWiM14}.  One practical alternative is to use
facts-on-the-ground in the form of signal strength measurements, instead of the prescriptive distance-based formula
\cite{us:MSWiM14,us:PODC14}.  To model that formally, the pessimistic reaction would be to replace the distance
assumption with an arbitrary signal-quality matrix, but that runs into the computational intractability monster, since
such a formulation can encode the coloring problem in general graphs \cite{GoussevskaiaHW14}.  A more moderate
approach is to relax the Euclidean assumption to more general metric spaces \cite{fangkeslinear}.  
The determinacy of the model is another issue.  To capture the probabilistic factors observed in the capture of
transmissions, one approach is to extend the basic SINR model accordingly, such as with Rayleigh fading; in that case,
it has been shown that applying algorithms based on the deterministic formula results in nearly equally good results in
that probabilistic setting \cite{damskesrayleigh}.

Our results suggest that a reassessment of the role of graphs as wireless models might be in order.
By paying a small factor (recalling, as well, that $\log^*(x) \le 5$ in this universe),
we can work at higher levels of abstraction, with all the algorithmic and analytical benefits that it accrues.
At the same time, hopes for fully constant-factor approximation algorithms for core scheduling problems have receded.
It remains to be seen what abstractions are possible for other related settings, especially the case of uniform power.

\mypara{Roadmap to the rest of the paper}
Following the basic definitions in Sec.~\ref{S:defs}, we derive from first principles what properties link conflict
graphs must satisfy (Sec.~\ref{S:formul}). This can be read independently of the rest of the paper.
We next derive (in Sec.~\ref{S:graphproperties}) two key properties of the family of graphs: how their chromatic numbers
relate and how their colorings can be approximated. We then introduce the definitions of the SINR model
before starting the technical core of the paper.
In Sec.~\ref{S:sandwich}, we show that feasibility is captured by two members of our conflict graph family. 
The implications are discussed in Sec.~\ref{S:mainresult}:
our main upper bound result, i.e. $O(\log^*{\Delta})$-approximation for {\scheduling} and {\wcapacity}; a necessary and sufficient condition for feasibility; and an explicit
polynomial-time computable measure of interference.
Limitation results are given in Sec.~\ref{S:limitations}, in particular 
that no better bounds are possible via conflict graphs.
For space reasons, only sketches of the proofs are given in this extended abstract.

\section{Definitions: Metrics, Functions, Graphs}
\label{S:defs}

\mypara{Communication Links}
Consider a set $L$ of $n$ \emph{links}\label{G:numlinks}, numbered from $1$ to $n$. Each link $i$ represents a
unit-demand communication request from a sender $s_i$ to a receiver $r_i$\label{G:siri} --- point-size wireless
transmitter/receivers (nodes) located in a metric space with distance function $d$\label{G:distance}.  We denote
$d_{ij}=d(s_i,r_j)$\label{G:asymdistance} and refer to $l_i=d(s_i,r_i)$\label{G:li} as the \emph{length} of link $i$.
We let $\Delta(L)$\label{G:delta} denote the ratio between the longest and the shortest link lengths in $L$, and drop
$L$ when clear from context.  We shall assume in the rest of the paper that all link lengths are distinct, which can be
achieved by arbitrarily (but consistently) breaking ties as needed.  For sets $S_1, S_2$ of links, we let $d(S_1,S_2)$
denote the minimum distance between \emph{a node} in $S_1$ and a node in $S_2$. In particular, we will extensively use
$d(i,j) = \min(d_{ij},d_{ji},d(s_i,s_j),d(r_i,r_j))$\label{G:symdistance}, the minimum distance between nodes on links $i,j$.

Let $S_i^+=\{j\in S : l_j> l_i\}$\label{G:liplus} denote the subset of links in a set $S$ that are longer than link $i$,
and similarly $S_i^-=\{j\in S : l_j< l_i\}$\label{G:liminus} the subset of links shorter than $i$.

\mypara{Doubling Metrics} The \emph{doubling dimension} of a metric space is the infimum of all numbers $\delta > 0$
such that for every $\epsilon$, $0 < \epsilon \le 1$, every ball of radius $r>0$ has at most $C\epsilon^{-\delta}$
points of mutual distance at least $\epsilon r$ where $C\geq 1$ is an absolute constant, and $0<\epsilon \leq 1$.
Metrics with finite doubling dimensions are said to be \emph{doubling}. For example, the $m$-dimensional Euclidean space
has doubling dimension $m$~\cite{heinonen}.  We let $m$\label{G:dimension} denote the doubling dimension of the space
containing the links.

\mypara{Functions}
A function $f$ is \emph{sub-linear} if $f(x) = O(x)$.
A function $f$ is \emph{strongly sub-linear} if for each constant $c\ge 1$, there is a constant $c'$ such that $cf(x)/x\le f(y)/y$ for all $x,y\ge 1$ with $x\ge c'y$. Note that if $f$ is strongly sub-linear then $f(x)=o(x)$. 
\\\textit{Examples.} The functions $f(x)=x^{1-\epsilon}$ for any constant $\epsilon>0$ and $f(x)=\log{x}$ are strongly
sub-linear\footnote{When not otherwise identified, logarithms are base 2.}, while $f(x)=x/\log{x}$ is not strongly sub-linear even though $f(x)=o(x)$.

Let $f$ be a strongly sub-linear function. For each integer $c\geq 1$, the function $f^{(c)}(x)$ is defined inductively by: $f^{(1)}(x)=f(x)$ and $f^{(c)}(x)=f(f^{(c-1)}(x))$\label{G:frepeated}. Let $x_0=\inf\{x\geq 1, f(x) < x\} +1$; such a point exists for any $f(x)=o(x)$. The function \emph{iterated $f$}, denoted $f^*(x)$\label{G:fstar}, is defined by:
\[
f^*(x)=\begin{cases}\min_c\{f^{(c)}(x)\le x_0\}, &\text{if }x> x_0,\\
1, &\text{otherwise.}
\end{cases}
\]
\textit{Examples. } For $f(x)=\log{x}$, $f^*(x)=\log^*{x}$ is the well known \emph{iterated logarithm}. It is also easy to check that for $f(x)=\log^{(c)}{x}$, $f^*(x)=\lceil\frac{\log^*{x}}{c}\rceil$ and for $f(x)=\sqrt{x}$, $f^*(x)=\lceil\log{\log{x}}\rceil$. Note also that $g^*(x)=\Theta(f^*(x))$ if $g(x)=cf(x)$ for a constant $c>0$.

\mypara{Graphs} For a graph $G$ and a vertex $v$, $N_{G}(v)$\label{G:nv} denotes the neighborhood of $v$ in $G$, i.e.,
the set of vertices adjacent with $v$.  $\chi(G)$\label{G:chi} denotes the chromatic number of graph $G$, the minimum
number of colors needed for a proper (vertex) coloring of $G$.

A \emph{$d$-inductive order} of a graph is an arrangement of the vertices from left to right such that each vertex has
at most $d$ \emph{post-neighbors}, or neighbors appearing to its right.  A \emph{$k$-simplicial elimination order} is
one where the post-neighbors of each vertex can be covered with $k$ cliques.  A graph is $d$-inductive ($k$-simplicial)
if it has a $d$-inductive ($k$-simplicial elimination) order.  It is well known that a $d$-inductive graphs are
$d+1$-colorable, while the coloring and weighted maximum independent set problems are $k$-approximable in $k$-simplicial
graphs \cite{ackoglu, kammertholey, yeborodin}.
The only inductive or simplicial order we consider for conflict graphs is the \emph{increasing order} of links by length.

\section{Formulations of Conflict Graphs}
\label{S:formul}

What kind of graphs are conflict graphs? By a ``conflict graph formulation'' we mean a
deterministic rule for forming graphs on top of a set of links.  For it to be meaningful as a general purpose
mechanism, such a formulation cannot be too context sensitive.  We shall postulate some axioms
(that by nature should be self-evident) that lead to a compact description of the space of
possible conflict graph formulations.
%%This is particularly important when we want to find  \emph{optimal} conflict graph representations.  

\begin{axiom}
A conflict graph formulation is defined in terms of the \emph{pairwise relationship} of links.
\label{axiom:pairwise}
\end{axiom}

By nature, graphs represent pairwise relationships; conflict graphs formulations are boolean predicates of
pairs of links.  More specifically, though, we expect the conflict graph to be defined in terms of the
relative standings of the link pairs. That is, the existence of an edge between link $i$ and link $j$ should
depend only on the properties of the two links, not on other links in the instance.  The only properties
of note are the ${4 \choose 2} = 6$ distances between the nodes in the links.

We refer to a \emph{conflict} between two links if the formulation specifies them to be adjacent in the
conflict graph; otherwise, they are \emph{conflict-free}.

\begin{axiom}
A conflict graph formulation is \emph{independent of positions and scale}. 
Translating distances or scaling them by a fixed factor does not change the conflict relationship.
\label{axiom:scale-free}
\end{axiom}

It is an essential feature of the SINR formula (that distinguishes it from other formulations, like unit-disc graphs) 
is that only relative distances matter.  Thus, the positions of
the nodes should not matter, only the pairwise distances, and only the relative factors among the
distances.  There is a practical limit to which links can truly grow, due to the ambient noise term.  However, that
only matters when lengths are very close to that limit; we will treat that case separately.

As a result of this axiom, we can factor out the length of the shorter of the two links considered.

\begin{axiom}
A conflict formulation is \emph{monotonic} with increasing distances.
\label{axiom:monotonicity}
\end{axiom}

The reasoning is that a conflict formulation should represent the degree of conflict between pairs of
links, or their relative ``nearness''.  Specifically, if two links conflict and their separation (i.e.,
one of the distances between endpoints on distinct links) decreases while the links stay of the same
length, then the links still conflict.  Similarly, if two links are conflict-free and the length of one
of them decreases (while their separation stays unchanged), the links stay conflict-free.

\begin{axiom}
A conflict formulation should respect pairwise incompatibility. That is, 
if two (links) cannot coexist in a feasible solution, they should be adjacent in the conflict graph.
\label{axiom:incompatible-pair}
\end{axiom}

\smallskip

In the case of conflict graphs for links in the SINR model with arbitrary power control, we propose an
additional axiom.

\begin{axiom}
  A conflict formulation for links under arbitrary power control is
  symmetric with respect to senders and receivers.
\label{axiom:symmetry}
\end{axiom}

Namely, it should not matter which endpoint of a link is the sender and which is the receiver when
determining conflicts. The key rationale for this comes from Kesselheim's sufficient condition for feasibility,
given here as Thm.~\ref{T:kesselheimconstant}. 
%Kesselheim used this symmetric rule to give a sufficient condition for feasibility; 
As we show in Sec.\ \ref{S:necsuf}, this formula is also a necessary condition in doubling metrics, up to constant factors.
Thus, feasibility is fully characterized by a symmetric rule (modulo constant factors).

As we shall see, the axioms and the properties of doubling metrics imply that only two distances really
matter in the formulation of conflict graphs: the length of the longer link, and the distance between the
nearest nodes on the two links (both scaled by the length of the shorter link). This motivates the
following definition. 

\begin{definition}
Let $f:\mathbb{R}_+\rightarrow \mathbb{R}_+$ be a positive non-decreasing sub-linear function. 
Two links $i,j$ are said to be \emph{$f$-independent} if
  \[ \frac{d(i,j)}{l_{min}} > f\left(\frac{l_{max}}{l_{min}}\right), \]
where $l_{min}=\min\{l_i,l_j\},l_{max}=\max\{l_i,l_j\}$, and otherwise they are \emph{$f$-adjacent}. 
A set of links is $f$-independent ($f$-adjacent) if they are pairwise $f$-independent ($f$-adjacent), respectively.

Let $L$ be a set of links. The conflict graph $\cG_f(L)$\label{G:gf} is the graph with vertex set $L$, where two vertices 
are adjacent if and only if they are $f$-adjacent.
\end{definition}

\noindent \emph{Remark.}
For the constant function $f(x)\equiv \gamma$ for a number $\gamma$, 
we use the notation $\cG_{\gamma}(L)$\label{G:ggamma} for the corresponding conflict graph.

We now argue that all conflict formulations satisfying the above axioms are essentially of the form
$\cG_f$, for a function $f$. They can differ from $\cG_f$ but not by more than what can be
accounted for by an appropriate constant factor. 

\begin{proposition}
Every conflict graph formulation $\cK$ is captured by $\cG_f$, for some non-decreasing and sub-linear function $f$.
Namely, there is a constant $\gamma$ such that 
$\cK$ is sandwiched by  $\cG_f$ and $\cG_{\gamma f}$, i.e.,
$\cG_f(L) \subseteq \cK(L) \subseteq \cG_{\gamma f}(L)$, for every link set $L$.
\label{prop:gf-suffices}
\end{proposition}

\begin{proof}
By Axiom \ref{axiom:pairwise}, a conflict formulation $\cK$ is a function of link pairs, more specifically, the
distances among the four points. By Axiom \ref{axiom:scale-free}, we can use normalized distances, and will 
choose to factor out the length of the shorter link.
By Axiom \ref{axiom:symmetry}, it does not matter which of them involve senders and which involve receivers.

Now, consider two links $i = (s_i,r_i)$ and $j = (s_j,r_j)$, where $l_i \le l_j$, and
assume without loss of generality that $s_i$ and $s_j$ are the nearest points on the two links.
Let us denote for short $d = d(i,j) = d(s_i,s_j)$.
We aim to show that decisions regarding adjacency in $\cK$
can be determined in terms of constant multiples of $d$ and $l_j$.

First, recall that by Axiom \ref{axiom:incompatible-pair}, pairwise incompatible links must be adjacent in any conflict
graph.  Thus, as derived in Thm.~\ref{T:lowerbound}, we may restrict attention to the case that $d(i,j) \ge c l_i$, for
an absolute constant $c$.  In that case, it follows that distance $d(r_i,s_j)$ is at most constant times the distance
from $i$ to $j$, i.e., $d(s_i,s_j) \le d(r_i,s_j) \le (1 + 1/c) d(s_i,s_j)$.

Next, we claim that $d(s_i,r_j)$ is within a constant multiple of $q = \max(d, l_j/2)$.
By definition of $d$, it holds that $d(s_i,r_j) \ge d$, while by the triangular inequality,
$d(s_i,r_j) \ge l_j - d(s_i,s_j) = l_j - d$. Thus, $d(s_i,r_j) \ge \max(d,l_j-d) \ge q$.
Also, by triangular inequality, $d(s_i,r_j) \le d + l_j \le 3q$. 

Finally, for by triangular inequality, $d(r_i,r_j) \le d + l_i + l_j \le d + 2l_j$,
and $d(r_i,r_j) \ge d - l_i - l_j \ge d - 2l_j$. So, defining $q' = \max(d, l_j/4)$,
we similarly obtain that $q' \le d(r_i,r_j) \le 9q'$.
It follows that all the five distances between endpoints are within constant multiples of $d(i,j)$ and $l_j$,
relative to the shorter link length $l_i$.

Hence, by monotonicity (Axiom \ref{axiom:monotonicity}), $\cK$ is dominated by a conflict graph
formulation $\cH$ defined by a monotone boolean predicate of two variables: length of the longer link
$l_j$, and the distance $d(i,j)$ between the links (scaled by the shorter link). 
But, an arbitrary monotone boolean predicate of two variables $x, y$ can be represented by a 
relationship of the form $y > f(x)$, for some function $f$. 
Thus, $\cK$ is dominated by $\cG_f$, for some non-decreasing function $f$.
Also, by the same arguments, $\cK$ dominates $\cG_{cf}$ for a constant $c$.

Finally, sub-linearity is a necessary condition, since super-linear growth would break Axiom \ref{axiom:monotonicity}. 
\end{proof}

\section{Properties of Conflict Graphs}\label{S:graphproperties}

We explore two types of properties of conflict graphs. The first type is concerned with gaps between the chromatic numbers of conflict graphs, or the relative difference of the chromatic numbers of graphs $\cG_{f}$ and $\cG_{f'}$.
We show that the introduction of $f$ increases the chromatic number of $\cG_{\gamma}$ by a rather small factor depending on $f$. This is a key result that will be used to derive the approximation factor in the main result of this paper. 
We also show that the introduction of constant factor $\gamma$ changes the chromatic number of $\cG_{f}$ only by a constant factor.

 In the second part we consider algorithmic properties of graphs $\cG_f$. In particular, we prove that graphs $\cG_f$ are constant-simplicial. Thus, constant factor approximation algorithms for vertex coloring, weighted maximum independent set and several other $\mathcal{NP}$-hard problems follow. This allows us to algorithmically approximate feasibility with graphs.

\subsec{Gaps Between Chromatic Numbers of Conflict Graphs}\label{S:gaps}
 We start by  showing that the difference between the chromatic numbers of $\cG_\gamma$ and $\cG_{\gamma f}$ is a factor of at most $O(f^*(\Delta))$, where $f^*$ is the iterated $f$ function. This result is obtained by proving that for any independent set $S$ in $\cG_{\gamma}(L)$, the graph $\cG_{\gamma f}(S)$ is $O(f^*(\Delta))$-inductive. 
 To this end, we want to show that, for any given link $i$ in $S$, any set $T$ of mutually $\gamma$-independent links in
 $S^+_i$ that are $\gamma f$-neighbors of $i$ is small, or $O(f^*(\Delta))$. We do so by showing that the progression of
 lengths of the links in $T$ must be fast growing, or inversely with $f$; the number of links must therefore be bounded
 by the iterated $f$ function.

\begin{theorem}\label{T:inductiveness}
For any set of links $L$, a constant $\gamma\geq 1$ and a non-decreasing strongly sub-linear function $f$, $\chi(\cG_{\gamma f}(L))=O(f^*(\Delta))\cdot \chi(\cG_{\gamma}(L))$. 
\end{theorem}
	\begin{proof}
	Let $S$ be a $\gamma$-independent set in $L$. Consider any link $i\in S$ and let $T$ denote the set of links in $S_i^+$ that are $f$-adjacent with $i$. We will show that $|T|= O(f^*(\Delta))$. 
	Note that for each $j\in T$, $d(i,j)\leq \gamma l_i f(l_j/l_i)$. Let $p_j$ denote the endpoint of link $j\in T$ closest to link $i$.  We split $T$ into two subsets $T_1$ and $T_2$, where
	\[
	T_1=\{j\in T : d(p_j,r_i) \leq \gamma l_i f(l_j/l_i)\}\mbox{ and }T_2=T\setminus T_1\subseteq\{j\in T : d(p_j,s_i) \leq \gamma l_i f(l_j/l_i)\}.
	\]
	Let us first consider $T_1$.
	Let $j,k\in T_1$ be two links with $l_{k}< l_j$. Then,
	\begin{align}
	\label{E:indineq1} d(p_j,r_i) &\leq \gamma l_i f(l_j/l_i),          &&\text{(because $j,k\in T_1$)}\\
	\label{E:indineq2} d(p_k,r_i) &\leq \gamma l_i f(l_k/l_i),        &&\\
	 \label{E:indineq3} d(p_j,p_k) &\geq d(j,k) > \gamma l_{k},     &&\text{($j,k$ are $\gamma$-independent)}
	\end{align}
	 By plugging the inequalities above into the triangle inequality $d(p_j,p_k) \leq d(p_j,r_i) + d(r_i,p_k)$, we obtain
	$
	\gamma l_{k} < \gamma l_i f(l_j/l_i) + \gamma l_i f(l_k/l_i)\leq 2\gamma l_i f(l_j/l_i),
	$
	where the last inequality follows from the assumption that $l_k< l_j$ and that $f$ is a non-decreasing function. Thus, 
	\begin{equation}\label{E:linkratio}
	 l_k/l_i < 2f(l_j/l_i).
	\end{equation}
	Denote $g(x)\equiv 2f(x)$. Note that $g(x)$ is strongly sub-linear; hence, there exists $x_0=\inf\{x\geq 1, g(x) < x\}+1$.
	Let $1,2,\dots,t=|T_1|$ be the arrangement of the links in $T_1$ in increasing order by length and let $\lambda_j=\frac{l_j}{l_i}$ for  $j=1,2,\dots,t$. Let $h$ be the link with the smallest index such that $\lambda_h \ge x_0$. We will bound the number of links in $A=\{1,2,\dots,h-1\}$ and $B=\{h,h+1,\dots,t\}$ separately. $|A|$ can be bounded by a simple application of the doubling property of the space. Note that for all $j\in A$, $f(l_j/l_i)\leq f(x_0)=O(1)$ because $l_j/l_i < x_0$. Thus, the system of inequalities (\ref{E:indineq1}-\ref{E:indineq3}) implies that the mutual distance between different points $p_j$ with $j\in A$ is at least $\gamma l_i$, while their distance from $r_i$ is at most $\gamma f(x_0) l_i$; hence, we have that $|A|=O(f(x_0)^m)= O(1)$.
	
	Now let us bound $|B|$, using~(\ref{E:linkratio}). We have that
	\[
	x_0\le \lambda_h < g(\lambda_{h+1})\le g(g(\lambda_{h+2}))\le \cdots\le g^{(t-h)}(\lambda_t),
	\]
	which implies that $t-h\le g^*(\lambda_t)= O(f^*(\Delta))$. Recall that $h=|A|=O(1)$; hence, $t=O(f^*(\Delta))$. 
	
	This completes the proof that $|T_1|=O(f^*(\Delta))$. The set $T_2$ is handled similarly. In this case, for any pair of links $j,k\in T_2$ with $l_j> l_k$, the system of inequalities (\ref{E:indineq1}-\ref{E:indineq3}) holds with $r_i$ replaced with $s_i$;
	the rest of the argument is identical. These results imply the theorem.
	\end{proof}

Next we show that the chromatic number of $\cG_{f}$ is of the same order as the chromatic number of $\cG_{\gamma f}$. We prove that any independent set $S$ in $\cG_{f}$ is constant-inductive as an induced subgraph of $\cG_{\gamma f}$. To this end, we prove, using the strong sub-linearity of $f$ (Lemma~\ref{L:shortneighbor}),
  that for any link $i\in S$, the set of neighbors of $i$ in $S_i^+$ mainly consists of links of length $\Theta(l_i)$. The number of such links is bounded by a rather straightforward application of the doubling property of the space, using the fact that those links form an independent set, while at the same time are adjacent with link $i$.

\begin{lemma}\label{L:shortneighbor}
Let $f$ be a non-decreasing strongly sub-linear function and $i,j,k$ be links. If links $j,k$ are longer than $i$, are $f$-independent and are $\gamma f$-adjacent with $i$, then $\min\{l_j,l_{k}\}\le cl_i$, where the constant $c$ depends only on function $f$ and constant $\gamma$.
\end{lemma}

\begin{proof}
Assume without loss of generality that $l_k< l_j$.
Since $f$ is strongly sub-linear, there is a constant $c>0$ such that $2\gamma f(x)/x\le f(y)/y$ whenever $x\ge cy$. We will show that $l_k\le c l_i$. 
Let us assume, for contradiction, $l_{k} >cl_i$. Let $p_j$ ($p_k$) denote the endpoint of link $j$ (link $k$) closest to link $i$. Then,
\begin{align}
\label{E:ii1}d(p_j,r_i) &\le  \gamma l_i f(l_j/l_i),                            &&\text{  ($j,k$ are $\gamma f$-adjacent with $i$)}\\
\label{E:ii2}d(p_k,r_i) &\le  \gamma l_i f(l_k/l_i),         &&\\
\label{E:ii3}d(p_j,p_k) &\geq d(j,k) > l_{k}f(l_j/l_k),        &&\text{  ($j,k$ are $f$-independent)}.
\end{align}
By plugging these inequalities into the triangle inequality we get:
\[
l_k f(l_j/l_k) <d(p_j,p_k)\leq d(p_j,r_i) + d(r_i,p_k) \leq  \gamma l_i f(l_j/l_i) + \gamma l_i f(l_k/l_i)\leq 2\gamma  l_i f(l_j/l_i).
\]
Let us denote $x=l_j/l_i$ and $y=l_j/l_k$. Note that $x>cy$. The inequality above asserts that $2\gamma f(x)/x>f(y)/y$, which contradicts the definition of $c$.  This completes the proof.
\end{proof}

\begin{theorem}\label{T:constantsens}
Let $L$ be a set of links and $f$ be a non-decreasing strongly sub-linear function. Then $\chi(\cG_{f}(L))=\Theta(\chi(\cG_{\gamma f}(L)))$ for any constant $\gamma>0$.
\end{theorem}

	\begin{proof}
	We assume w.l.o.g. that $\gamma\geq 1$. Thus, we only need to show that $\chi(\cG_{\gamma f}(L))=O(\chi(\cG_{f}(L)))$. Consider any independent set $S$ in $\cG_f(L)$. It suffices to show that $\cG_{\gamma f}(S)$ is constant-inductive.  Let us fix a link $i$ and let $T=S_i^+\cap N_{\cG_{\gamma f}(S)}(i)$ denote the set of neighbors of $i$ in $\cG_{\gamma f}(S)$ that are longer than $i$. It is enough to show that $T=O(1)$. Recall that all the links in $T$ are $\gamma f$-adjacent with $i$ and are $f$-independent among each other. By applying Lemma~\ref{L:shortneighbor} for all pairs $j,k\in T$, we conclude that there is a constant $c$ such that all the links in $T$, except perhaps one, have length at most $cl_i$. Let $T'$ be the subset of $T$ containing those links. It is enough to show that $|T'|=O(1)$. For each link $j\in T'$, let $p_j$ denote the endpoint of $j$ closest to link $i$. We split $T'$ into two subsets $T'_1$ and $T'_2$, where
	\[
	T'_1=\{j\in T' : d(p_j,r_i) \leq \gamma l_i f(l_j/l_i)\}\mbox{ and }T'_2=T'\setminus T'_1\subseteq\{j\in T' : d(p_j,s_i) \leq \gamma l_i f(l_j/l_i)\}.
	\]
	We first bound $|T'_1|$. Note that for each pair of links $j,k\in T'_1$ with $l_j> l_k$, the distance $d(p_j,p_k)$ is at least $d(p_j,p_k)>l_kf(l_j/l_k)\ge f(1) l_i$ because $j,k$ are $f$-independent. On the other hand, for each $j\in T'_1$, the distance $d(p_j,r_i)$ is at most $d(p_j,r_i)\le \gamma l_i f(l_j/l_i)\le \gamma f(c) l_i$ because $i,j$ are $\gamma f$-adjacent and $l_j\le cl_i$. We conclude that $|T'_1|=O((\gamma f(c)/f(1))^m)=O(1)$, using the doubling property of the metric space.
	
	We can prove that $|T'_2|=O(1)$ in a similar manner, by replacing $r_i$ with $s_i$ in the formulas. These results imply the theorem.
	\end{proof}

\subsec{Algorithmic Properties of Conflict Graphs}

We prove that every conflict graph $\cG_f$ with strongly sub-linear function $f$ is constant-simplicial. This guarantees, among other properties, that the vertex coloring and maximum weighted independent set problems in these graphs can be efficiently approximated within constant factors~\cite{ackoglu, kammertholey, yeborodin}. 

We give a simple argument that holds in the plane, where it holds for essentially all sub-linear functions.
It is based on splitting the plane into $60\degree$ sectors emanating from a given node of a link, 
and arguing that all adjacent longer links within a sector must form a clique.
With a more detailed argument, the result can be extended to general doubling metrics, but requires then strong sub-linearity.
\begin{proposition}\label{P:2dperfectness}
Let $f$ be a non-decreasing function such that $f(x)/x$ is non-increasing and let $L$ be a set of links in the plane. Then $\cG_f(L)$ is 12-simplicial.
\end{proposition}

\begin{proof}
  Consider a link $i$ and let $S\subseteq L_i^+$ be the longer neighbors of $i$ in $\cG_f(L)$. We partition $S$ into two
  sets: $S_1$, with links that are closer to the sender node $s_i$, and $S_2$, the links that are closer to
  $r_i$. Consider $S_1$ first. For a link $j\in S_1$, let $p_j$ denote the endpoint of $j$ that is the closest to $s_i$,
  i.e., $d(i,j)=d(p_j,r_i)$. Then we have, from $f$-dependence, that $d(p_j,s_i)\le l_i\cdot f(l_j/l_i)$ for all $j\in
  S_1$. Partition the plane into six $60\degree$ sectors emanating from $s_i$. Let $j,k$ be two links such that
  $l_k>l_j$ and $p_j$ and $p_k$ fall in the same sector. By considering the triangle $(p_j,p_k,s_i)$, we can see that
  the edge $p_jp_k$ must be no longer than $\max\{d(p_j,s_i),d(p_k,s_i)\}$. Thus,
\[
d(j,k)\le d(p_j,p_k)\le l_i\cdot f(l_k/l_i)\le l_j\cdot f(l_k/l_j),
\]
where the second inequality follows because $f(x)$ is non-decreasing and the third one follows because $f(x)/x$ is non-increasing. This shows that the links $j,k$ must be adjacent. Thus, $S_1$ can be covered with at most 6 cliques. An identical argument holds for $S_2$, by splitting the plane around $r_i$. This completes the proof.
\end{proof}

\noindent \emph{Remark.}
The proof above cannot be replicated in general doubling metrics. Moreover, it can be shown that there are functions $f$ such that $\cG_f$ is constant-simplicial in the plane, but is not so even in a one dimensional doubling metric. The linear function $f(x)=x$ is such an example. This claim can be demonstrated by the following example with $n+1$ points $p_0,p_1,\dots,p_n$, where $d(p_i,p_0) =2^i$ for all $i\geq 1$ and $d(p_i,p_k)=2^i+1$ for all $i>k\geq 1$. It is straightforwardly checked that the points in this example form a $1$-dimensional doubling space. On the other hand, the $f$-neighborhood of point $p_0$ contains a set of $\log{n}$ $f$-independent points. A  similar example was also observed in~\cite{Welzl08} in a different context. However, we prove that $\cG_f$ is constant-simplicial in doubling metrics for \emph{strongly sub-linear} functions $f$.

\medskip

We show as before, that for a link $i$,  the set of longer links adjacent with $i$ can be covered with a small number of cliques. We start by showing (using Lemma~\ref{L:shortneighbor}) that the neighbors of $i$ of length $\Omega(l_i)$ must form a single clique. The rest of the links (having length $O(l_i)$) can be covered by cliques using a simple clustering procedure. We show that the cluster-heads  must be separated by a distance $\Omega(l_i)$, which implies, using the doubling property of the space and the fact that all the links in consideration are neighbors of link $i$ and have length $O(l_i)$, that the number of clusters (cliques) must be bounded by a constant.

\begin{theorem}\label{T:perfectness}
Let $f$ be a non-decreasing strongly sub-linear function with $f(x)\ge 1$ for all $x\ge 1$. Then for each set $L$, $\cG_{f}(L)$ is constant-simplicial.
\end{theorem}
\begin{proof}
Let $i\in L$ be any link and let the subset $S\subseteq L_i^+$ consist of links in $L$ longer than $i$ and $f$-adjacent with $i$. We show that $S$ can be covered with a constant number of cliques.

Let $j,k\in S$ be any pair of $f$-independent links in $S$. By Lemma~\ref{L:shortneighbor}, there is a constant $c>0$ depending only on $f$, such that at least one of the links $j,k$ is not longer than $cl_i$. Thus, the set of links $j\in S$ with $l_j> cl_i$ forms a single clique.

It remains to show that the subset of $S$ with links of length at most $cl_i$ can be covered with a constant number of cliques. Let $T$ denote this subset. 

We split $T$ into a set of cliques by the following procedure. Pick an arbitrary link $j\in T$. Let $N^1_{j}$ denote the set of links $k$ (including link $j$) with $\min\{d(s_{j},s_k),d(s_j,r_k)\}\leq l_i/2$ and $N^2_j$ denote the set of links $k$ with $\min\{d(r_{j},s_k),d(r_j,r_k)\}\leq l_i/2$. Set $T=T\setminus (N^1_{j}\cup N^2_{j})$. Repeat until $T$ is empty. Let $R$ be the links picked by the procedure above. 

Note that for each $j\in R$, $N^1_j$ and $N^2_j$ are cliques. Indeed, consider e.g. $N_j^1$ and let $k_1,k_2\in N_j^1$ and $l_{k_1}< l_{k_2}$. The triangle inequality and the assumption $f(x)\geq 1$ for $x\geq 1$ imply that $d(k_1,k_2)\leq l_i< l_{k_1}f(l_{k_2}/l_{k_1})$, which means that $k_1,k_2$ are $f$-adjacent. Moreover, these cliques cover $T$. Let us show that $|R|= O(1)$. For each $j\in R$, let $p_j$ denote the endpoint of $j$ that is closest to $i$. We split $R$ into two subsets, $R_1$ and $R_2$ (based on the fact that $R$ consists of links $f$-adjacent with $i$), where
\[
R_1=\{j\in R : d(p_j,r_i) \leq l_i f(l_j/l_i)\}\mbox{ and }R_2=R\setminus R_1\subseteq\{j\in R : d(p_j,s_i) \leq l_i f(l_j/l_i)\}.
\]
Let us consider $R_1$ first. Recall that for each $j\in R$, $l_j\leq cl_i$. From the definition of $R_1$ we have that for each $j\in R_1$, $d(p_j,r_i)\leq l_if(l_j/l_i)\leq l_if(c)$, i.e., the points $p_j$ are inside the ball of radius $l_if(c)$ centered at $s_i$. On the other hand, we have by the construction of $R$ that for any $j,k\in R$, $d(p_j,p_k)\geq d(j,k)>l_i/2$. Hence, by applying the doubling property of the metric space we get that $|R_1|\leq (l_if(c)/(l_i/2))^m=(2f(c))^m=O(1)$. 
A similar argument holds for $R_2$, by replacing $r_i$ with $s_i$. Thus, we get $|R|=O(1)$, which completes the proof.
\end{proof}

\section{Definitions: SINR, Feasibility}\label{S:sinr}

\begin{table}
\centering
    \begin{tabular}{  l  l c l }
		\textit{Notation} & \textit{Meaning} & \textit{Topic} & \textit{Page}\\\hline
    $f^{(c)}(x)$ & function $f$ applied $c$ times & & {\pageref{G:frepeated}}\\
    $f^*(x)$ & iterated $f$ & \textit{Functions} & \pageref{G:fstar}\\
		$\tlog(x)$ & $\approx \log^{2/(\alpha-m)}(x)$ & & \pageref{G:tlog}\\\hline
		$m$ & the doubling dimension of the metric space & \textit{Metric Space} & \pageref{G:dimension}\\
    $d$ & the distance function of the metric space & & \pageref{G:distance}\\\hline
		$n$ & the number of links & & \pageref{G:numlinks} \\ 
    $s_i,r_i$ & sender and receiver nodes of link $i$ & & \pageref{G:siri} \\ 
    $l_i$ & the length of link $i$, $l_i=d(s_i,r_i)$ & & \pageref{G:li}\\ 
    $d_{ij}$ & the distance from $s_i$ to $r_j$, $d_{ij}=d(s_i,r_j)$ & & \pageref{G:asymdistance} \\
    $d(i,j)$ & the minimum distance between links $i,j$ & \textit{Links} & \pageref{G:symdistance}\\
    $\Delta(L)$ & the maximum ratio between link lengths in $L$ & & \pageref{G:delta} \\
    $L_i^+$ & the subset of links of $L$ longer than link $i$ & & \pageref{G:liplus} \\
    $L_i^-$ & the subset of links of $L$ shorter than link $i$ & & \pageref{G:liminus} \\\hline
    $\chi(G)$ & the chromatic number of graph $G$ & & \pageref{G:chi}\\
		$N_{G}(v)$ & the neighborhood of vertex $v$ in graph $G$ & {\textit{Graphs}} & \pageref{G:nv}\\
    $\cG_f(L)$ & the $f$-adjacency graph over the set $L$ &  & \pageref{G:gf} \\
    $\cG_\gamma(L)$ & the $f$-adjacency graph over $L$ with $f(x)=\gamma$ & & \pageref{G:ggamma}\\\hline
    $P$ & power assignment, $P:L\rightarrow \mathbb{R}_+$ & & \pageref{G:power}\\
    $\alpha$ & the path loss exponent & & \pageref{G:alpha}\\
    $\beta$ & the SINR threshold value & & \pageref{G:beta}\\
    $N$ & the ambient noise term & \textit{SINR} & \pageref{G:noise}\\
    $OPTS(L)$ & the optimum schedule length of set $L$ & & \pageref{G:opts}\\
    $I$ & the influence operator & & \pageref{G:influence}\\
    $I(L)$ & same as $\max_{i\in L}{I(L_i^-,i)}$ & & \pageref{G:il}\\
    \end{tabular}
    \caption{Notations.}
\end{table}

\mypara{SINR Model and Feasibility}
A \emph{power assignment} for a set $L$ of links is a function $P:L\rightarrow \mathbb{R}_+$. For each link $i$, $P(i)$\label{G:power} defines the power level used by the sender node $s_i$. 
In the \emph{physical model (or SINR model)} of communication~\cite{rappaport}, when using a power assignment $P$, a transmission of a link $i$ is successful if and only if 
\begin{equation}\label{E:sinr}
\frac{P(i)}{l_i^{\alpha}}\geq \beta\cdot \left(\sum_{j\in S\setminus \{i\}}\frac{P(j)}{d_{ji}^{\alpha}} + N\right),
\end{equation}
where $N$\label{G:noise} is a constant denoting the ambient noise, $\beta$\label{G:beta} denotes the minimum SINR (Signal to Interference and Noise Ratio) required for a message to be successfully received, $\alpha \in (2,6)$\label{G:alpha} is the path loss constant and $S$ is the set of links transmitting concurrently with link $i$. Here the left side of the inequality is interpreted as the received signal power of link $i$ and the sum on the right side is interpreted as the interference on link $i$ caused by concurrently transmitting links.

A set $S$ of links is called $P$-\emph{feasible} if the condition~(\ref{E:sinr}) holds for each link $i\in S$ when using power $P$. We say $S$ is \emph{feasible} if there exists a power assignment $P$ for which $S$ is $P$-feasible. Similarly, a collection of sets is $P$-feasible/feasible if each set in the collection is.
Note that we do not assume limits on the available power, which means that the noise term can be ignored.
The case of a maximum power limit requires primarily that the links that are close to maximum length be handled separately using the maximum power available \cite{kesselheimflexible}, something that remains to be studied.

\mypara{The Influence Operator and a Sufficient Condition for Feasibility} 
The \emph{influence operator} $I$\label{G:influence} is defined as follows. For links $i,j$, let
$I(i,j)=\frac{l_i^\alpha}{d(i,j)^\alpha}$ and define $I(i,i)=0$ for simplicity of notation.  The operator $I$ is
additively expanded: for a set $S$ of links and a link $i$, let $I(S,i)=\sum_{j\in S}I(j,i)$ and $I(i,S)=\sum_{j\in
  S}I(i,j)$.  We will use the notation $I(L)=\max_{i\in L}{I(L_i^-,i)}$\label{G:il}.

In order to identify feasible sets, we will use the following sufficient condition for feasibility.
\begin{theorem}\cite{kesselheimconstantfactor}\label{T:kesselheimconstant}
For any set of links $L$ in a metric space, if $I(L)<\frac{1}{2\cdot 3^{\alpha}(4\beta+2)}$, then $L$ is feasible.
\end{theorem}

\mypara{Sensitivity of Feasible Sets} A set of links is called $p$-$P$-feasible if it is $P$-feasible with the parameter $\beta$ replaced with number $p$. The following sensitivity argument has proved useful. It shows, in particular, that constant factor changes to the threshold parameter $\beta$ do not affect asymptotic results by more than a constant factor.

 \begin{theorem}\cite{HB14}\label{T:signalstrengthening}
Let $p$, $p'$ be positive values, $P$ be a power assignment,  and  $L$ be a $p$-$P$-feasible set.
Then $L$ can be partitioned into $\left\lceil 2p'/p\right\rceil$ sets each of which is $p'$-$P$-feasible.
 \end{theorem}

\mypara{Fading Metrics}
\emph{Fading metrics} are doubling metrics with doubling dimension $m < \alpha$. We shall assume,
without stating so explicitly, that the links are located in a fading metric.

\section{Capturing Feasibility with Conflict Graphs}\label{S:sandwich}

We show that for appropriate constant $\gamma>0$ and function $f$, SINR-feasibility is ``trapped'' between graph representations $\cG_{\gamma}$ and $\cG_{f}$; namely, each  feasible set is an independent set in $\cG_{\gamma}(L)$  and each independent set in $\cG_f(L)$ is feasible. 
 In particular, this holds for $f(x)=\gamma'\tlog(x)$ for an appropriate constant $\gamma'>0$, where the function $\tlog(x)$\label{G:tlog} is defined for $x\ge 1$ by
$\tlog(x)=\max(\log^{2/(\alpha - m)}(x), 1)$.
The gap between these approximations is quantified using our results in Sec.\ \ref{S:gaps}, ultimately leading to $O(\log^*\Delta)$ approximation for scheduling problems.

\subsec{Independence of Feasible Sets}
The theorem below is based on the simple observation that two links in the same ``highly feasible'' set must be spatially separated by at least a multiple of the length of the shorter link, implying $\gamma$-independence for some $\gamma>0$. The constant $\gamma$ may then be adapted using Thm.\ \ref{T:signalstrengthening}, i.e. a feasible set can be split into a constant number of $\gamma'$-independent sets for any constant $\gamma'>0$.
\begin{theorem}\label{T:lowerbound}
For any constant $\gamma > 0$, a $(\gamma+1)^\alpha$-feasible set is $\gamma$-independent. In particular, if $\beta>1$ then each feasible set is $(\beta^{1/\alpha}-1)$-independent.
\end{theorem}

\begin{proof}
It suffices to show that two links in the same $(\gamma+1)^\alpha$-feasible set must be $\gamma$-independent. Let $i,j$ be such links. Since $i,j$ are in the same $(\gamma+1)^\alpha$-feasible set, the SINR condition implies that there is a power assignment $P$ such that:
\[
P(i)/l_i^\alpha > (\gamma+1)^\alpha P(j)/d_{ji}^\alpha\text{ and }P(j)/l_j^\alpha > (\gamma+1)^\alpha P(i)/d_{ij}^\alpha.
\]
By multiplying together the inequalities above, canceling $P(i)$ and $P(j)$ and raising to the power of $1/\alpha$, we obtain: 
\begin{equation}\label{E:independence}
d_{ij}d_{ji} > (\gamma+1)^2l_il_j.
\end{equation}
 Let us show first that $\min\{d_{ij},d_{ji}\} > \gamma \min\{l_i,l_j\}$. Indeed, if the opposite was true, e.g. if $d_{ij} \leq \gamma\min\{l_i,l_j\}$, the the triangle inequality would imply that $d_{ji}\leq d_{ij} + l_i + l_j\leq (\gamma+2)\max\{l_i,l_j\}$, which would contradict to (\ref{E:independence}): 
$
d_{ij}d_{ji}\leq \gamma(\gamma+2)l_il_j\leq (\gamma+1)^2l_il_j.
$

Now consider $d(s_i,s_j)$. Let us assume, for contradiction, that e.g.  $d(s_i,s_j)\le \gamma l_i \le \gamma l_j$. Then the triangle inequality would imply $d_{ji} \leq d(s_i,s_j) + l_i\leq (\gamma+1)l_i$ and $d_{ij}\leq d(s_i,s_j) + l_j\leq (\gamma+1)l_j$, which would again yield a contradiction to (\ref{E:independence}). We  prove in the same manner that $d(r_i,r_j) > \gamma\min\{l_i,l_j\}$ and conclude that $d(i,j) > \gamma \min\{l_i,l_j\}$, i.e., $i$ and $j$ are $\gamma$-independent.
\end{proof}

\subsec{Feasibility of Independent Sets}
Here we show that for a large enough constant $\gamma> 0$, $\gamma\tlog$-independence implies feasibility. In particular, we show that if a set $S$ is $\gamma\tlog$-independent then $I(S)= O(\gamma^{m-\alpha})$.
Since we assumed that $m<\alpha$, an appropriate choice of $\gamma$ yields feasibility via Thm.\ \ref{T:kesselheimconstant}.

The argument consists of the following stages. For any given link $i\in S$, we first split $S_i^-$ into length classes, or \emph{equilength subsets}, where each equilength subset contains links differing by at most a factor of 2 in length. We bound the influence on link $i$ for each of those subsets separately, and then combine those bounds using the additivity of the influence operator $I$. 

For each equilength subset $S$ the following common technique is applied: partition the plane into concentric annuli around the link $i$, count the number of links in each annulus and bound $I(S_i^-,i)$ based on these numbers and the fact that the links within the same annulus have almost the same influence on link $i$ (because they are at roughly the same distance from $i$ and have roughly similar lengths). The number of links in each annulus can be bounded using the doubling property of the space and independence of the links. The influence bound obtained for each subset $S$ is $O((\gamma \tlog(l_i/\ell))^{m-\alpha})$, where $\ell$ is the longest link length in $S$. The function $\tlog$ is chosen so that combining those bounds in a sum results in an upper bound of $I(S_i^-,i) = O(\gamma^{m-\alpha})$.

We will use the following two technical observations.
 \begin{fact}\label{L:convex}
 Let $\alpha\geq 1$ and $r\geq 0$ be real numbers. Then
 $
\frac{1}{ r^\alpha} - \frac{1}{(r+1)^\alpha}\leq \frac{\alpha}{(r+1)^{\alpha+1}}.
 $
 \end{fact}

\begin{fact}\label{L:integral}
Let $g(x) = \displaystyle\frac{1}{(q + x)^{\delta}}$, where $\delta > 1$ and $q\geq 1$. Then 
$
\sum_{r=0}^\infty {g(r)}= \displaystyle O\left(q^{1-\delta}\right).
$
\end{fact}

The following lemma bounds the influence of an equilength $1$-independent set $S$ on a long link $i$ that is $f$-independent from the set $S$. This will be the main building block to be used for showing that $\gamma \tlog$-independent sets are feasible. The proof uses the annuli argument mentioned above.

\begin{lemma}\label{L:garbage}
 Let $f$ be a non-decreasing function, such that $f(x)\geq 1$ whenever $x\geq 1$. Let $S$ be an equilength  $1$-independent set of links, and let $i$ be a link s.t. for each $j\in S$, $l_i\geq l_j$ and $i$ and $j$ are $f$-independent. 
Then
 $
  I(S,i) = O\left((f(l_i/\ell))^{m - \alpha} \right),
 $
where $\ell$ denotes the longest link length in $S$.
 \end{lemma}

\begin{proof}
Let us denote $q=f(l_i/\ell)$.  Note that $q\geq 1$ because $l_i/\ell\ge 1$. 

 Let us split $S$ into two subsets $S'$ and $S''$, where $S'$ contains the links of $S$ that are closer to $r_i$ than to $s_i$,
i.e., $S'=\{j\in S: \min\{d(s_j,r_i), d(r_j,r_i)\} \leq \min\{d(s_j,s_i), d(r_j,s_i)\}\}$ and $S''=S\setminus S'$. Let us bound $I(S',i)$ first.

For a link $j\in S'$, let $p_j$ denote the endpoint of link $j$ that is closest to $r_i$, i.e., $d(i,j)=d(p_j,r_i)$. 
Consider the ``chain'' of subsets $S_1\subseteq S_2\subseteq \dots\subseteq S'$, where 
 \[
 S_r=\{j\in S': d(j,i)=d(p_j,r_i)\leq q \ell/2+(r-1)\ell/2\}.
 \]
Let $M_r\geq\max_{j\in S_r\setminus S_{r-1}}{I(j,i)}$ be some upper bound on the maximum of $I(j,i)$ in the annulus $S_{r}\setminus S_{r-1}$ for $r=2,3,\dots$. The value $I(S',i)$ can be bounded as follows:
 \begin{align}
\nonumber I(S',i) &= I(S_1,i) + \sum_{r\geq 2}{\sum_{j\in S_r\setminus S_{r-1}}{I(j,i)}}\\
\nonumber &\leq I(S_1,i) + \sum_{r\geq 2}{M_r \cdot|S_r\setminus S_{r-1}|}\\
\nonumber &=I(S_1,i) + \sum_{r\geq 2}{M_r(|S_r|- |S_{r-1}|)}\\
\label{E:gupperbound} &=I(S_1,i) - |S_1|M_2 + \sum_{r\geq 2}{|S_r|(M_r - M_{r+1})},
 \end{align}
 where the last line follows by a simple rearrangement of the sum. We will next bound the sizes of subsets $S_r$ and find bounds $M_r$. 
 \begin{claim}\label{C:gs1}
 $S_1=\emptyset.$
 \end{claim}
 \begin{proof}
For each link $j\in S'$, $d(p_j,r_i)=d(i,j) > l_j q\geq \ell q/2$ because $i$ and $j$ are $f$-independent and $S$ is an equilength set with maximum link length $\ell$ and minimum link length at least $\ell/2$. 
 \end{proof}
 \begin{claim}\label{C:gs2}
 For each $r\geq 2$, $|S_r|\leq C \left(q+r-1\right)^{m}$, where $C$ is an absolute constant.
 \end{claim}
 \begin{proof}
 We bound $|S_r|$ using the doubling property of the metric space. Consider any $j,k\in S_r$ such that $l_j\geq l_k$. By the assumption, $j,k$ are  $1$-independent; hence,
 $ d(p_j,p_k) \geq d(j,k)> \min\{l_j,l_k\}\geq \ell/2. $
 By the definition of $S_r$, $d(p_j,r_i)\leq q \ell/2+(r-1)\ell/2$ for each $j\in S_r$. Because the metric space has doubling dimension $m$, the number of points $p_j$ with $j\in S_r$ (hence, also the size $|S_r|$) can be bounded as follows: 
\[
|S_r|=|\{p_j\}_{j\in S_r}|< C\cdot \left(\frac{ q\ell/2+(r-1)\ell/2 }{\ell/2}\right)^{m} = C \left(q+r-1\right)^{m}.
 \]
 \end{proof}
\begin{claim}\label{C:gmr}
Let $M_r=\frac{2^\alpha}{\left(q+r-2\right)^\alpha}$. For each $r\geq 2$, $\max_{j\in S_{r}\setminus S_{r-1}}\{I(j,i)\} < M_r$.
\end{claim}
\begin{proof}
For each $r>1$ and for any link $j\in S_r\setminus S_{r-1}$, we have that $l_j \leq \ell$ and $d(i,j) > q\ell/2+(r-2)\ell/2$; hence, 
$
I (j, i) = \frac{l_j^{\alpha}}{d(i,j)^\alpha}
 < \left(\frac{\ell}{q\ell/2+(r-2)\ell/2}\right)^\alpha
 =\frac{2^\alpha}{\left(q+r-2\right)^\alpha}.
$
\end{proof}
By Claim~\ref{C:gs1}, the first two terms of (\ref{E:gupperbound}) are zero. 
Let us fix any $r\geq 2$. Let $M_r$ be as in Claim~\ref{C:gmr}. By Fact~\ref{L:convex}, 
$
M_r - M_{r+1}\leq \alpha2^{\alpha}/(q+r-1)^{\alpha+1},
$
and by Claim~\ref{C:gs2},
\[
|S_r|(M_r - M_{r+1})<\frac{C\alpha2^{\alpha}(q+r-1)^m}{(q+r-1)^{\alpha+1}}=\frac{C\alpha2^{\alpha}}{(q+r-1)^{\alpha - m +1}}.
\]
 By plugging these inequalities into (\ref{E:gupperbound}) and using Fact~\ref{L:integral}, we get the desired bound for $I(S',i)$:
\begin{align*}
I(S',i) <\sum_{r\geq 2}{|S_r|(M_r - M_{r+1})}< C\alpha2^{\alpha}\sum_{r\geq 2}{\frac{1}{(q+r-1)^{\alpha-m+1}}}\in O(q^{m-\alpha}).
\end{align*}
The proof holds symmetrically for the set $S''$. Recall that $S''$ consists of the links of $S$ that are closer to the sender $s_i$ than to the receiver $r_i$. Now, we can define the set $\{p_j\}_{j\in S''}$ where $p_j$ is the endpoint of link $j$ that is closest to $r_i$, for each $j\in S''$. The rest of the proof will be identical, by replacing $r_i$ with $s_i$ in the formulas.
\end{proof}

 Having a bound for the influence of each equilength set, we can now split the whole set into equilength subsets (length classes), bound the influence of each equilength subset using Lemma~\ref{L:garbage} and combine them into a series that converges when we choose $f(x)=\gamma\tlog(x)$.
  \begin{theorem}\label{T:main}
   Let $L$ be a $\gamma\tlog$-independent set with $\gamma\ge 1$. 
Then
 $
 I(L) = O\left(\gamma^{m-\alpha} \right).
 $
 \end{theorem}

  \begin{proof}
  Let us fix an arbitrary link $i\in L$.
We have for each $j\in L_i^-$,
$d(i,j) > \gamma l_{j}\tlog(l_i/l_j)$ because of $\gamma\tlog$-independence and that $l_i\ge l_j$. Let $\ell_0$ denote the minimum link length in $L_i^-$. 
We partition $L_i^-$ into at most $\lceil\log{l_i/\ell_0}\rceil$ equilength subsets $L_1, L_2,\dots$ as follows: 
\[
L_t=\{j\in L_i^-: 2^{t-1}\ell_0 \leq l_j<2^t \ell_0\},
\]
 for $t=1,2,\dots$. Let $\ell_t$ be the longest link length in $L_t$.
The conditions of Lemma~\ref{L:garbage} hold for each $L_t$: it is an equilength $1$-independent set ($\gamma\tlog$-independence implies $1$-independence for $\gamma\ge 1$) and is $f$-independent from link $i$, with $f=\gamma\tlog$. Note also that $f(x)\geq 1$ when $x\geq 1$. Applying the lemma, we obtain
\[
I(L_t,i) = O\left((\gamma\tlog(l_i/\ell_t))^{m-\alpha}\right).
\]
Let $d$ denote the largest index $t$ for which $L_t$ is not empty. By the definition of function $\tlog$ we have that $\tlog(l_i/\ell_d)\geq 1$ and for each $t<d$, $\tlog(l_i/\ell_t)=\log^{2/(\alpha-m)}(l_i/\ell_t)\geq (d-t)^{2/(\alpha-m)}$. Thus, 
\[
I(L_i^-,i)= \sum_{t=1}^{d}{I(L_t,i)}\leq c\gamma^{m-\alpha}\left(1+ \sum_{t=1}^d{\left((d-t)^{2/(\alpha-m)}\right)^{m-\alpha}}\right)=O(\gamma^{m-\alpha}),
\]
where $c$ is a constant. Since this holds for arbitrary $i\in L$, we have that $I(L)=O(\gamma^{m-\alpha})$.
 \end{proof}

Since the theorem above holds for any $\gamma\geq 1$, we 
obtain the desired result.
\begin{corollary}\label{C:mainresult}
There is a constant $\gamma\geq 1$ such that each $\gamma\tlog$-independent set is feasible.
\end{corollary}

\section{Implications}\label{S:mainresult}

\subsec{{\scheduling} and {\wcapacity} Approximation}
Using our method of capturing feasibility with graphs, we approximate {\scheduling} and {\wcapacity} problems within a factor of $O(\log^*{\Delta})$. Let us first formally define the problems and related terms.

A \emph{schedule} for a set $L$ of links is a partition of $L$ into feasible subsets (or \emph{slots}). 
The \emph{length} of the schedule is its number of slots.
The {\scheduling} \emph{problem} is to find a minimum length schedule for a given set $L$. The length of an optimal schedule for $L$ is denoted $OPTS(L)$\label{G:opts}.

 The {\wcapacity} problem is the generalized dual of {\scheduling}, where given a set $L$ of links  with weights $\omega:L\rightarrow \mathbb{R}^+$, the goal is to find a feasible subset $S\subseteq L$ of maximum weight $\sum_{i\in S}{\omega(i)}$.

\begin{theorem}
There are polynomial $O(\log^*{\Delta})$-approximation algorithms for {\scheduling} and {\wcapacity}. The approximation is obtained by 
coloring the graph $\cG_{\gamma\tlog}$ (for an appropriate constant $\gamma\ge 1$) in the case of {\scheduling} and by approximating its maximum weighted independent set in the case of {\wcapacity}.
\end{theorem}

\begin{proof}
First consider the {\scheduling} problem. Let $L$ be an input to {\scheduling}.  We construct and color the graph $\cG_{\gamma\tlog}(L)$ with constant
$\gamma$ chosen as in Corollary~\ref{C:mainresult}. By Corollary~\ref{C:mainresult}, such a coloring corresponds to
a feasible schedule.

To derive the approximation factor, observe on one hand that in view of Thms.\ \ref{T:signalstrengthening}
and~\ref{T:lowerbound}, any schedule of $L$ can be refined into a coloring of $\cG_\gamma(L)$ with only constant factor
increase in the number of slots. Thus, $OPTS(L)=\Omega(\chi(\cG_{\gamma}(L)))$.  On the other hand, by Thm.\
\ref{T:inductiveness}, $\chi(\cG_{\gamma \tlog}(L))=O(\tlog^*(\Delta))\cdot \chi(\cG_{\gamma}(L))
%%% =O(\log^*{\Delta})\cdot \chi(\cG_{\gamma}(L))
=O(\log^*{\Delta})\cdot OPTS(L)$.  It is readily verified that the function $\gamma\tlog$ is strongly sub-linear,
implying, via Thm.\ \ref{T:perfectness}, that $\cG_{\gamma\tlog}(L)$ is
constant-simplicial and thus colorable within constant approximation factor.

Now consider the {\wcapacity} problem. Let a set $L$ be given. As in the case of {\scheduling}, we first construct the graph $\cG_{\gamma\tlog}(L)$ with constant $\gamma$ chosen as in Corollary~\ref{C:mainresult}. We find a constant-factor approximate weighted maximum independent set in $\cG_{\gamma\tlog}(L)$ using the fact that this graph is constant-simplicial. By Corollary~\ref{C:mainresult}, the resulting set is feasible, i.e. it is a valid solution for {\wcapacity}. Now let us derive the approximation factor. 
Let $W_l$ and $W_u$ be the weights of the weighted maximum independent sets in $\cG_{\gamma}(L)$ and $\cG_{\gamma\tlog}(L)$ respectively, and let $W_o$ be the weight of the optimal solution to {\wcapacity} in $L$. Let $S$ be a solution to {\wcapacity} in $L$. Since $S$ is feasible,  it can be split into a constant number of $\gamma$-independent subsets, by Thms.\ \ref{T:signalstrengthening} and~\ref{T:lowerbound}. Let $S'$ be the largest weight subset. Obviously, the weight of $S'$ is  $\Omega(W_o)$, implying that $W_l =\Omega(W_o)$, as $S'$ is an independent set in $\cG_{\gamma}$. On the other hand, Thm.\ \ref{T:inductiveness} implies that $S'$ can be refined into  at  most $O(\log^*{\Delta})$ $\gamma\tlog$-independent subsets. The largest weight subset will have weight at least $\Omega(W_l/\log^*{\Delta})$, which implies that $W_u=\Omega(W_l/\log^*{\Delta})=\Omega(W_o/\log^*{\Delta})$.
\end{proof}

\subsec{Measure of Interference}
%\subsec{Explicit Bounds on the Optimum Scheduling Number}
%
While approximation algorithms give bounds relative to an optimal value, it is frequently advantageous to have bounds in
terms of some intrinsic parameters or more easily computable properties. Thus the interest in bounding chromatic numbers
of graphs in terms of clique numbers, broadcast algorithms in terms of network diameter, and routing time in terms of
``congestion + dilation''. Our results also imply bounds for the optimum schedule length that can be efficiently
computed from the network topology.  
Previous such results involved logarithmic factors in $n$ and/or $\Delta$ \cite{fangkeslinear,kesvokdistributed}.

Let $G$ be a $k$-simplicial graph and let $v_1,v_2,\dots,v_n$ be a $k$-simplicial elimination order of vertices, which
for our conflict graphs is by increasing link length.  A $k$-approximate coloring of $G$ is obtained by coloring the
vertices greedily in reverse order.  
The number of colors used is at most the maximum post-degree plus 1, or
$\max_i\{|N(v_i)\cap\{v_{i+1},\dots,v_n\}|\}+1\le k\cdot \chi(G)+1$. We therefore define
  \[ B_f(L) = \max_{i\in L}|\{j\in L : l_j\geq l_i, d(i,j)\le l_if(l_j/l_i)\}|, \]
for a function $f$, and observe that $\chi(\cG_f(L)) = \Theta(B_f(L))$.
The results of Sec.\ \ref{S:sandwich} and \ref{S:limitations} then imply the following theorem.

\begin{theorem}
There are constants $a,b>0$ and $\gamma\ge 1$, such that for any set $L$,
\[a\cdot B_{\gamma}(L) \le OPTS(L) \le b\cdot B_{\tlog}(L)\text{ and }\frac{B_{\tlog}(L)}{B_{\gamma}(L)}=O(\log^*{\Delta(L)}).\]
Moreover, there are infinitely many instances $L'$ and $L''$ s.t. $\frac{OPTS(L')}{B_{\gamma}(L')}=\Omega(\log^*{\Delta(L')})$ and $\frac{B_{\tlog}(L'')}{OPTS(L'')}=\Omega(\log^*{\Delta(L'')})$.
\end{theorem}

\subsec{A Necessary and Sufficient Condition for Feasibility}
\label{S:necsuf}
Another interesting implication of Thm.\ \ref{T:main} is the following result that shows that the sufficient condition for feasibility stated in Thm.\ \ref{T:kesselheimconstant} is essentially necessary in doubling metric spaces. 
This result is of independent interest, as it may prove useful for improved analysis of various problems.
It should be noted that this theorem does not hold in general metric spaces.

The proof consists of two parts, bounding the influence on a link $i$ by faraway links (i.e., links that are highly independent from link $i$) on one hand using Thm.\ \ref{T:main}, and by near links (the rest) on the other hand, using simple manipulations of the SINR condition.
\begin{theorem}\label{T:necessary}
Let $L$ be a $3^\alpha$-feasible set of links. Then,
$
  I(L) = O(1).
 $
 \end{theorem}

\begin{proof}
Let us fix a link $i\in L$ and denote $S=L_i^-$.
We split $S$ into two subsets $S_1$ and $S_2$, where for each link $j\in S_1$, $j$ and $i$ are $f$-independent with $f(x)=2x$, and $S_2=S\setminus S_1$. 

Recall that by Thm.\ \ref{T:lowerbound}, $3^{\alpha}$-feasibility implies $2$-independence of $S_1$. The bound $I(S_1,i)= O(1)$ then follows by applying an analogue of Thm.\ \ref{T:main} with $\gamma=1$ and with $f$-independence instead of $\tlog$-independence, which can be done because $\tlog(x)= O(f(x))$.

It remains to show that $I(S_2,i)= O(1)$. Let $P$ be a power assignment for which $L$ is $P$-feasible.
Then, the SINR condition gives us the following inequalities:
\[
\frac{P(i)}{l_i^\alpha} > 3^{\alpha}\sum_{j\in S_2}{\frac{P(j)}{d_{ji}^\alpha}},\mbox{ and }\frac{P(j)}{l_j^\alpha} >  3^{\alpha}\frac{P(i)}{d_{ij}^\alpha}\mbox{ for all }j\in S_2.
\]
 By replacing $P(j)$ with $3^{\alpha}\frac{P(i)l_j^{\alpha}}{d_{ij}^\alpha}$ in the first inequality and simplifying the expression, we get:
\begin{equation}\label{E:equation2}
\sum_{j\in S_2}\frac{l_i^\alpha l_j^\alpha}{d_{ij}^\alpha d_{ji}^\alpha} \le 9^{-\alpha}.
\end{equation}
In order to extract a bound on $I(S_2,i)$ from (\ref{E:equation2}), we will show that one of the values $d_{ij}^{\alpha},d_{ji}^{\alpha}$ in the denominator can be canceled out with $l_i^{\alpha}$ in the numerator and the other one can be replaced with $d(i,j)^\alpha$ by only introducing additional constant factors in the expression. Such a modification will transform the left side of (\ref{E:equation2}) into $I(S_2,i)$.

Let us assume w.l.o.g. that $d_{ij}\geq d_{ji}$. Recall that for each $j\in S_2$, $d(i,j) \leq 2l_i$ by definition of $S_2$. Using the triangle inequality, we obtain $d_{ij}\leq d(i,j) + l_i+l_j \leq 4l_i$. On the other hand, as it was mentioned above, the set $S_2$ is $2$-independent, which implies that $d_{ji}\geq d(i,j) > 2l_j$. Using the triangle inequality again, we obtain:
\[
d(s_i,s_j)\geq d_{ij} - l_j\geq d_{ji} - l_j > d_{ji}/2\text{ and }d(r_i,r_j)\geq d_{ji} - l_j>d_{ji}/2.
\]
Thus, $d(i,j) > d_{ji}/2$. By replacing $d_{ij}$ with $4l_i$ and $d_{ji}$ with $2d(i,j)$ in the left-hand part of (\ref{E:equation2}), we obtain the desired bound: $I(S_2,i)\le (8/9)^\alpha$. Since this holds for an arbitrary $i\in L$, we get that $I(L)=O(1)$. 
\end{proof}

\noindent \emph{Remark.}
Note that Thm.\ \ref{T:signalstrengthening} implies that any feasible set can be refined into a constant number of $3^\alpha$-feasible subsets. Thus, the influence function fully captures feasibility in fading metrics, modulo constant factors.

\section{Limitations of the Graph-Based Approach}\label{S:limitations}

We have found that conflict graphs can achieve a remarkably good, yet super-constant, approximation for scheduling problems in doubling metrics. We examine in this section how far this approach can be pushed, obtaining essentially tight bounds. We treat these issues in terms of the {\scheduling} problem.

In the first part of the section, we expose the limitations of the graph method in Euclidean spaces. We show, in particular, that conflict graphs do not yield any non-trivial approximation to the {\scheduling} problem in terms of the number of links $n$. In particular, they cannot lead to constant factor approximation. We also consider approximation limits in terms of the parameter $\Delta$, and show that for all reasonable functions $f$, the approximation factor is at least $\Omega(\log^*{\Delta})$. Thus, the approximation factor we obtained cannot be improved within a conflict-graph framework. 
Note that the instances we construct are embedded on the real line, i.e., in one dimensional space. 

In the second part of the section, we find that the graph method cannot provide any non-trivial approximation guarantees
in general metric spaces, neither in terms of $n$ nor $\Delta$.

\subsec{Euclidean Spaces}
In the following theorem, we construct, for any function $f=\omega(1)$, a feasible set of $f$-adjacent links. The
construction is based on the following observations. On the one hand, it follows from Thm.\ \ref{T:kesselheimconstant}
that any set of exponentially growing links arranged sequentially by the order of length on the real line is (almost)
feasible. On the other hand, given such a set $S$ of links on the line, a new link $j$ can be formed so that
$j$ is $f$-adjacent to all the links in $S$ while the set $S\cup j$ stays feasible; the only requirement is that $j$
be long enough. Our construction then builds recursively on these ideas.
\begin{theorem}\label{T:ndependence}
Let $f(x)=\omega(1)$. For any integer $n > 0$, there is a \emph{feasible} set $L$ of $n$ links arranged on the real line, such that $\cG_f(L)$ is a clique, i.e., $\chi(\cG_f(L))=n$. Moreover, if $f(x)\ge g(x)$ ($x\ge 1$) for a strongly sub-linear increasing function $g(x)$ with $g(x)=\omega(1)$, then $n = \Omega(g^*(\Delta))$.
\end{theorem}

\begin{proof}
Consider a set of $cn$ links $\{1,2,\dots,cn\}$ arranged sequentially from left to right on the real line, where $c>0$ is a constant to be chosen later. Each link $i$ is directed from left to right
and for each $i=1,2,3,\dots,n-1$, the nodes $s_{i+1}$ and $r_{i}$ share the same location on the line,  i.e., $r_i=s_i+l_i=s_{i+1}$. See Figure~\ref{fig:notcomplicated}.
 The lengths of links are defined inductively, as follows. We set $l_1=1$, and for $i\ge 1$, we choose $l_{i+1}$ to be the minimum value satisfying:
\begin{align}
\label{E:feasibility} l_{i+1}&\ge 2l_i \\ 
\label{E:clique}2d(i+1,j)=2d_{i+1,j}&\le l_jf(l_{i+1}/l_j)\mbox{ for all } j\leq i.
\end{align}
Such a value of $l_{i+1}$ can be chosen as follows. By the inductive hypothesis, we have $l_j\geq 2l_{j-1}$ for $j=2,3,\dots,i$. This implies that $l_i\ge \sum_{j=1}^{i-1}{l_j}$. Then, we have that $d_{i+1,j}=\sum_{t=j+1}^i{l_t}\le 2l_i$ for $j=1,2,\dots,i$. Thus, it is enough to choose $l_{i+1}$ so that $l_{i+1}\ge 2l_i$ and $4l_i \leq l_jf(l_{i+1}/l_j)$, which can be done using $f=\omega(1)$ and the fact that the values of $l_j$ for $j=1,2,\dots,i$ are already fixed at this point. This completes the construction.
\begin{figure}[htbp]
\begin{center}
\includegraphics[width=0.8\textwidth]{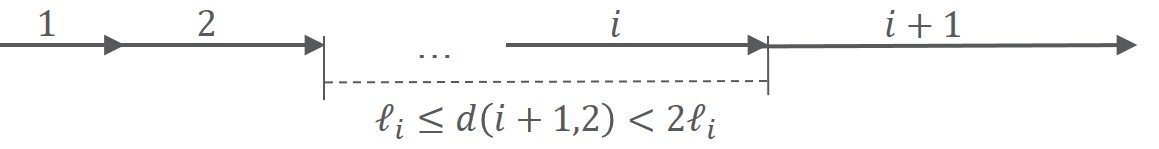}
\caption{The construction in Thm.\ \ref{T:ndependence}.}
\label{fig:notcomplicated}
\end{center}
\end{figure}
First note that (\ref{E:clique}) implies that $\cG_f(L)$ is a clique. It remains to argue feasibility. Consider the odd numbered links $S=\{1,3,\dots,...,2t+1\}$. Let us fix a link $2k+1\in S$. Note that for each $j\in S_{2k+1}^-$, $d(j,2k+1)\geq l_{2k}$. We have that
\[
I(S_{2k+1}^-,2k+1)=\sum_{j\in S_{2k+1}^-}{\frac{l_j^{\alpha}}{d(j,2k+1)^{\alpha}}}\leq \sum_{j\in S_{2k+1}^-}{\left(\frac{l_j}{l_{2k}}\right)^{\alpha}}\leq \sum_{j\in S_{2k+1}^-}{\frac{l_j}{l_{2k}}}\leq 1,
\]
where the second inequality holds because $l_j/l_{2k}\leq 1$ and the last inequality follows from (\ref{E:feasibility}). Thus, we can extract a constant fraction $S'$ of $S$ that is feasible, using Thm.\ \ref{T:signalstrengthening}. With the right choice of the constant $c$ in the beginning of the proof we  have that $|S'|=n$. This proves the first part of the theorem.

Now let us assume that $f(x)\ge g(x)$ for a strongly sub-linear function $g(x)$ with $g(x)=\omega(1)$. Then, there is a constant $x_0$ such that $g(x) < x$ for all $x \ge x_0$ (because $g(x)=o(x)$) and there is a constant $c$ such that $2g(x)/x\le g(y)/y$ whenever $x\ge cy$ (strong sub-linearity). In this case we repeat the construction above with slight modifications. 

We set $l_1=1$ and set $l_{i+1}>\max\{c,x_0\}$  be the minimum value s.t. $g(l_{i+1}) \ge 2l_i$, for $i=1,2,\dots$ (such a value exists because $g(x)=\omega(1)$). Let us show that the conditions (\ref{E:feasibility}-\ref{E:clique}) hold with these lengths.
  
Since $l_{i+1}\ge x_0$, we have that $l_{i+1} > g(l_{i+1}) \ge 2l_i$, which implies (\ref{E:feasibility}). This in turn implies, as observed in the first part of the proof, that $d(i+1,j) < 2l_i$ for all $2\le j\le i$. Let us denote $x=l_{i+1}/l_1=l_{i+1}$ and $y=l_{i+1}/l_j$. Note that $x/y=l_j \ge c$, so we have, by strong sub-linearity of $g$,  that $g(y)/y \ge 2g(x)/x$, or equivalently, that $l_j\cdot g(l_{i+1}/l_{j}) \ge 2\cdot g(l_{i+1})$; hence $l_j\cdot g(l_{i+1}/l_j) \ge 4l_i > 2d(i+1,j)$ for all $2\le j \le i$, which means that (\ref{E:clique}) also holds. 

It remains to prove the lower bound for $n$. Recall that the value of $l_{i+1}$ is the minimum satisfying $g(l_{i+1})\ge 2l_i$ for $i=1,2,\dots,n-1$. Then, we have $g(l_{i+1}/2) < 2l_i$ or, equivalently, $h(l_{i+1}/2) < l_i/2$, where $h(x)=g(x)/4$. Thus,
\[
1/2=l_1/2 > h(l_2/2) > h(h(l_3/2))> \dots > h^{(n-1)}(l_n/2)=h^{(n-1)}(\Delta/2),
\]
  which implies that $n =\Omega( h^*(\Delta/2))=\Omega(g^*(\Delta))$.
\end{proof}

\begin{corollary}
In terms of the number of links $n$, the approximation factor for {\scheduling} when using $\cG_f$ with any $f=\omega(1)$ is no better than $n$.
\end{corollary}

By choosing $g(x)=\gamma\tlog(x)$ in Thm.\ \ref{T:ndependence}, we obtain that the approximation factor of $O(\log^*{\Delta})$ cannot be improved for $\cG_{\gamma\tlog}$.

\begin{corollary}
Let $f(x)=\Omega(\log^{(c)}x)$ for a constant $c$. Then, for each $\Delta>0$, there is a \emph{feasible} set of links $L$ with $\Delta(L)=\Omega(\Delta)$, such that $\cG_{f}(L)$ is a clique of size $\Theta(\log^*{\Delta(L)})$.
\end{corollary}

While the theorem above shows that graphs $\cG_f$ with $f=\Omega(\log^{(c)}x)$ for some constant $c$ require too much
separation,
the theorem below shows that graphs $\cG_f$ with $f=O(\log^{1/\alpha}x)$ provide insufficient separation, leading,
perhaps surprisingly, to a similar sized gap of $\log^*{\Delta}$. Namely,
$\chi(\cG_f(L'))=\frac{OPTS(L')}{\Omega(\log^*{\Delta})}$ holds for certain instances $L'$. The construction 
follows the general structure of Thm.\ 7 in~\cite{halmitSODA12} of a lower bound for scheduling the edges of a minimum spanning tree of a set of points in the plane.
There are two technical challenges to overcome, in order to implement this
structure in our setting. First, the construction of~\cite{halmitSODA12} is not $f$-independent. Second, even when
ignoring the $f$-independence requirement, the lower bound for the scheduling number obtained in~\cite{halmitSODA12} is
only $\Omega(\log{\log^*{\Delta}})$. 

\begin{theorem}\label{T:hardinstance}
Let $f(x)=O(\log^{1/\alpha}x)$. For each $\Delta>0$, there is an $f$-independent set of links $L$ on the real line with $\Delta(L)=\Omega(\Delta)$ that cannot be scheduled in fewer than $\Theta(\log^*{\Delta(L)})$ slots.
\end{theorem}
We describe the idea of the construction informally. The construction is inductive, starting from a trivial instance
$L_1$ containing a single link. For $t \ge 1$, assume there is an instance $L_t$ having the desired properties, i.e.,
$L_t$ is $f$-independent and with $OPTS(L) \ge t$.
In order to construct the instance $L_{t+1}$, consider a single link $j$ that is longer than the links in $L_t$ and
place it at distance $d$ from $L_t$ so that
all the links in $L_t$ are $f$-independent from $j$. Let $I_0$ denote the minimum influence of a link from $L_t$ on link
$j$. Now, take $k$ identical copies of $L_t$ and place them at a distance $d$ from $j$. This will of course violate the
independence between different instances, which we will address shortly, but they will still be independent from link
$j$. The idea is that if the number of copies $k$ is large enough, then for any set $S$ containing at least one link
from each copy, we will have $I(S,j)=kI_0>c_0$, where $c_0$ is a constant large enough to ensure that $S\cup\{j\}$ is
infeasible (based on Thm.\ \ref{T:necessary}). This will mean that any schedule of the link $j$ and the $k$ copies 
must place at least \emph{one whole copy} of $L_t$ in slots separate from $j$. Since it takes at least $t$ slots to
schedule one copy of $L_t$, it takes at least $t+1$ slots to schedule
all the copies together with link $j$.

It remains to address the issue of $f$-independence between different copies. 
Note that because of the scale-invariance of the influence operator, we can scale a copy of $L_t$ by a factor $s$ and place it further than before, at a distance $s\cdot d$ from link $j$ and still have the minimum influence of $I_0$ on $j$. However, in order for this influence to be taken into the account, the link $j$ must still be longer than the links in the scaled instance. Using this trick, we can scale different copies by different factors and guarantee their mutual independence, while preserving the properties we had in the case of identical copies. Since the link lengths must grow exponentially at each step $t$, the number $t$ of slots required will be small compared to the number of links and the parameter $\Delta$, but will still be $\Omega(\log^*{\Delta})$.

\begin{proof}
For a set $S$ of links, we will use $diam(S)$ to denote the diameter of $S$, or the maximum distance between nodes in $S$.

We will construct a set of links that cannot be scheduled in fewer than $\Theta(\log^*{\Delta})$ $3^\alpha$-feasible slots, relying on the necessary condition for feasibility (Thm.\ \ref{T:necessary}). This will be sufficient to prove the theorem, as Thm.\ \ref{T:signalstrengthening} will imply that there cannot exist a $\beta$-feasible schedule with $\Theta(\log^*{\Delta})$ slots for that set, for any constant $\beta$.

  Let us fix a function $f$. Note that since $f=O(\log^{1/\alpha})$, there is a constant $C\geq 1$ s.t. $f(x)\leq C\log^{1/\alpha}x$.  We construct sets  $L_t$ of links recursively. The construction is illustrated in Figure~\ref{fig:complicated}. All the links will be arranged on the real line and the receiver of
  each link will be to the right of the sender.  Initially, we have a set $L_1$ consisting of a single
  link of length $1$, for which a single slot is sufficient and necessary.  Suppose that we have already
  constructed $L_{t}$ with the property that at least $t$ slots are required for scheduling $L_{t}$. The
  instance $L_{t+1}$ is constructed as follows using $k$ scaled copies of $L_t$, where $k$ is to be
  determined. First we place a single very long link $j_{t+1}$ in the line. We then add, in order from left to
  right, copies 
  $L_t^1,L_t^2,\dots,L_t^{k}$ of $L_t$ to the right of $j_{t+1}$, where $L_t^s$ is the copy of $L_t$ scaled by a factor
  $8^s$. The idea is to make the construction so that the following properties hold: 
  \begin{enumerate}[(i)]
  \setlength{\itemsep}{0cm}%
  \setlength{\parskip}{0cm}%
\item   {$L_{t+1}$ is $f$-independent,}\label{EN:independence}
  \item {$t=\Omega(\log^*{\Delta(L_t)})$,} \label{EN:lowerbound}
  \item {for any set $S=\{i_1,i_2,\dots,i_k\}$
  with $i_s\in L_t^s$, $s=1,2,\dots,k$, we have that$I(S,j_{t+1})>c_0$ for a constant $c_0$ of our choice.}\label{EN:inconsistency2}
 \end{enumerate}
  The last property ensures that each $3^\alpha$-feasible schedule of $L_{t+1}$ must put a whole copy $L_t^s$ in a slot separate from $j_{t+1}$. Indeed, if there was a schedule that placed at least one link from each copy $L_t^s$ in the same slot with $j_{t+1}$ then we would get a contradiction with (\ref{EN:inconsistency2}): we would have $I(S,j_{t+1}) =O(1)$ for some $S$ as above, due to Thm.\ \ref{T:necessary}. Recall that $L_t$ needs at least $t$ slots to be scheduled, and so does each copy of it. It follows that $L_{t+1}$ needs at least $t+1=\Omega(\log^*{\Delta(L_t)})$ slots to be scheduled, one for $j_t$ and at least $t$ for scheduling the copies of $L_t$. Proving the properties (\ref{EN:independence}-\ref{EN:inconsistency2})  will complete the proof of the theorem.
\begin{figure}[htbp]
\begin{center}
\includegraphics[width=0.85\textwidth]{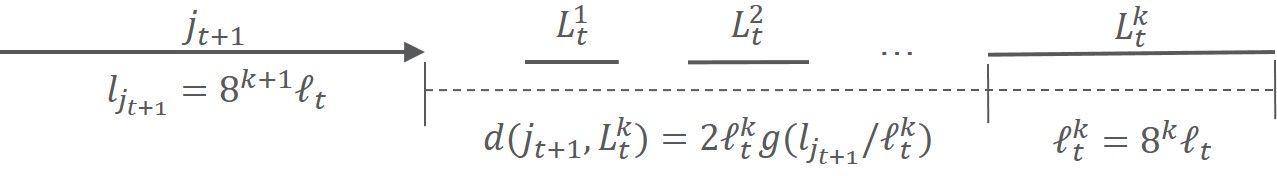}
\caption{The recursive construction of $L_{t+1}$.}
\label{fig:complicated}
\end{center}
\end{figure}

Now let us describe the inductive step of the construction in detail. Let $\ell_t=diam(L_t)$ denote the diameter of $L_t$. The number of copies of $L_t$ is $k=2^{c\ell_t}$ for a large enough constant $c$. The length of link $j_{t+1}$ is set to $l_{j_{t+1}}=8^{k+1}\ell_t$. It remains to specify the placement of each copy $L_t^s$ so as to guarantee the desired properties of $L_{t+1}$.

We assume by induction that the links within each copy of $L_t$ are $f$-independent. We place the copies $L_t^s$ so that the links between any two copies are $f$-independent and are $f$-independent from $j_{t+1}$. Let $\ell_t^{s}=diam(L_t^s)=8^s\ell_t$ denote the diameter of $L_t^s$. 
Let $g(x)=C\log^{1/\alpha}x$. We place each copy $L_t^s$ at a distance $d(L_t^s,j_{t+1})=2\ell_t^sg(l_{j_{t+1}}/\ell_t^s)$ from $j_{t+1}$. The construction is ready.

We first prove the property (\ref{EN:independence}).
\begin{claim}
With the distances defined as above, the set $L_{t+1}$ is $f$-independent.
\end{claim}
\begin{proof}
Consider any link $i\in L_t^s$. We have that 
\[
d(i,j_{t+1})\geq d(L_t^s,j_{t+1})=2\ell_t^sg(l_{j_{t+1}}/\ell_t^s)\geq 2l_ig(l_{j_{t+1}}/l_i)\geq 2l_if(l_{j_{t+1}}/l_i),
\]
where the second inequality follows from the fact that $xg(c/x)$ is an increasing function of $x$ and that $l_i<\ell_t^s$, and the third inequality follows because $f(x)\leq g(x)$ for all $x$. Thus, all the links in $L_t^s$ are $f$-independent from $j_{t+1}$. Now let us show that any two links $i,k$ with $l_i\leq l_k$ from different copies $L_t^s$ and $L_t^r$ with $s > r$ are $f$-independent (no matter which link is from which copy). Since $f(x)\leq g(x)$, it will be enough to show that
\begin{equation}\label{E:pessdist}
d(i,k)> l_ig(l_k/l_i).
\end{equation}
 Recall that $xg(c/x)$ is an increasing function of $x$. Then, for a fixed $k$, the right side of (\ref{E:pessdist}) is maximized when $l_i$ is maximum. On the other hand, for a fixed $i$,  the value $g(l_k/l_i)$ is maximized when $l_k$ is maximum, because $g$ is an increasing function. Let $j_t$ denote the maximum length link in $L_t$. Then, the maximum link length in $L_t^s$ (in $L_t^r$) is $8^sl_{j_{t}}$ ($8^rl_{j_{t}}$). Therefore, it is enough to show that 
\[
d(i,k)>\ell_t^rg(8^sl_{j_t}/(8^rl_{j_t}))=\ell_t^rg(8^{s-r})=C(3(s-r))^{1/\alpha}\ell_t^r.
\]
We have that
\[
d(i,k) \geq d(L_t^s,L_t^r)=d(L_t^s,j_{t+1}) - d(L_t^r,j_{t+1}) - \ell_t^r\geq 2\ell_t^sg(l_{j_{t+1}}/\ell_t^s) - 3\ell_t^rg(l_{j_{t+1}}/\ell_t^r).
\]
The term $g(l_{j_{t+1}}/\ell_t^r)$ can be bounded by
\[
g(l_{j_{t+1}}/\ell_t^r)=g(8^{s-r}l_{j_{t+1}}/\ell_t^s)\leq 3^{s-r}g(l_{j_{t+1}}/\ell_t^s),
\]
where the last inequality follows because $g(8x)\leq 3 g(x)$ for $x\geq 2$ (note that $\alpha \geq 1$). Thus,
\[
d(i,k) \geq 2\ell_t^sg(l_{j_{t+1}}/\ell_t^s) - 3^{s-r+1}\ell_t^rg(l_{j_{t+1}}/\ell_t^s)>C(2\cdot 8^{s-r} - 3\cdot 3^{s-r})\ell_t^r>C(3(s-r))^{1/\alpha}\ell_t^r.
\]
\end{proof}
Next, we can observe that (the first line follows because the links are arranged linearly) 
\begin{align}
\nonumber \ell_{t+1}&=l_{j_{t+1}}+d(L_t^k,j_{t+1}) + \ell_t^k\\
\nonumber &\leq l_{j_{t+1}}+2\ell_t^kg(l_{j_{t+1}}/\ell_t^k) + \ell_t^k\\
\nonumber &=8^{k+1}\ell_t + 8^k\ell_tg(8) + 8^k\ell_t\\
&=O(8^{ 2^{c\ell_t}}).
\end{align}
Since the minimum link-length in $L_{t+1}$ is $1$, we can conclude that $\Delta(L_{t})<\ell_t\leq 2\uparrow (c_1t)$ for a constant $c_1$ and for each $t$, where $\uparrow$ denotes the tower function.
This implies that $t=\Omega(\log^*{\Delta(L_t)})$. The property (\ref{EN:lowerbound}) is now proven.

It remains to check that (\ref{EN:inconsistency2}) holds. Let us consider a link $i_s$ from $L_{t}^s$ where $i_s$ is the copy of link $i$ in $L_{t}$. We have that
\[
d(i_s,j_t)\leq \ell_t^s+d(L_t^s,j_t)= \ell_t^s + 2C\ell_t^s\log^{1/\alpha}{(l_{j_{t+1}}/\ell_t^s)}\leq c_2\ell_t^s(k-s+1)^{1/\alpha},
\]
for a constant $c_2$. This implies:
\[
I(i_s,j_{t+1})=\left(\frac{l_{i_s}}{d(i_s,j_{t+1})}\right)^\alpha \geq \left(\frac{l_{i_s}}{  c_2(k-s+1)^{1/\alpha}\ell_{t}^s} \right)^\alpha \geq \frac{1}{c_3(k-s)\ell_{t-1}},
\]
where we used the fact that $l_{i_s}/\ell_t^s=l_i/\ell_t\geq 1/\ell_t$.
Now, let $i_s,$ $s=1,2,\dots,k$ be a set of links where $i_s\in L_{t}^s$ and they are not necessarily the copies of the same link of $L_{t}$. Then, 
\[
I(S,j_{t+1}) = \sum_{s=1}^{k}I(i_s,j_{t+1}) > \sum_{1}^{k}{\frac{1}{c_3(k-s+1)\ell_{t}}} = \Omega\left(\frac{\log{k}}{\ell_{t}}\right).
\]
Recall that $k=2^{c\ell_t}$. By taking the constant $c$ large enough, we can thus guarantee the property (\ref{EN:inconsistency2}). This completes the proof of all the properties of $L_t$ and the proof of the theorem.
\end{proof}

\subsec{General Metric Spaces}
The following theorem shows that conflict graphs can be arbitrarily far from schedules in general metric spaces. Given a function $f$, the construction consists of an $f$-independent set of \emph{unit length} links. Since all links have length $1$, $f$-independence is equivalent to $f(1)$-independence. The separation between the links is just enough to ensure $f(1)$-independence. However, since all the links are equally ($f(1)$-) separated from any given link, their interference accumulates and only a constant number of links can be scheduled in the same slot. This leads to schedules of length $\Theta(n)$. 
% Note also that $\Delta=1$ for this set of links.
\begin{proposition}
For each function $f$ and any $n\geq 1$, there is an $f$-independent set of $n$ unit length links (hence, $\Delta=1$) that cannot be scheduled into less than $\Theta(n)$ slots.
\end{proposition}

\begin{proof}
Let $L=\{1,2,\dots,n\}$ be the set of links. We define the lengths and the distances between the links such as to ensure the metric constraints hold. For each link $i$ we define $l_i=1$. The distances between the nodes are defined as follows:
\begin{itemize} 
\setlength{\itemsep}{0cm}%
  \setlength{\parskip}{0cm}%
\item[\checkmark]{sender to sender: $d(s_i,s_j)=f(1)\cdot (l_i+l_j)=2f(1)$,}
 \item[\checkmark]{sender to receiver: $d(s_i,r_j)=d(s_i,s_j)+l_j=2f(1)+1$,}
 \item[\checkmark]{ receiver to receiver: $d(r_i,r_j)=d(s_i,s_j)+l_i+l_j = 2f(1)+2$.}
\end{itemize}
It is straightforward to check that such distances define a metric. Moreover, the whole set of links in this metric is $f$-independent, since $d(i,j)>f(1)\cdot l_i=l_if(l_j/l_i)$. Let us consider any $P$-feasible subset $S$ of $k$ links for a power assignment $P$. Let us fix a link $i\in S$. The SINR condition implies: $P(i) > \beta\sum_{j\in S\setminus\{i\}}\frac{P(j)l_{i}^\alpha}{d_{ji}^\alpha}$ and $P(j) > \beta \frac{P(i)l_j^\alpha}{d_{ij}^\alpha}$ for all $j\in S\setminus \{i\}$.  By replacing $P(j)$ with $\frac{\beta P(i)l_j^\alpha}{d_{ij}^\alpha}$ in the first inequality and canceling the term $P(i)$, we obtain:
\[
1 >\beta^2\sum_{j\in S\setminus\{i\}}{\frac{l_i^{\alpha }l_j^\alpha}{d_{ij}^\alpha d_{ji}^\alpha}}=\beta^2\sum_{j\in S\setminus\{i\}}{\frac{1}{(2f(1)+1)^2}}=\frac{\beta^2(|S|-1)}{(2f(1)+1)^2},
\]
which implies that $|S| < \left(\frac{2f(1)+1}{\beta}\right)^2+1=O(1)$. Since $S$ was an arbitrary feasible subset of $L$, we conclude that $L$ cannot be split into less than $\Theta(n)$ feasible subsets.
\end{proof}

\mypara{Acknowledgements} We thank Christian Konrad and Jaikumar Radhakrishnan for helpful suggestions.

\bibliographystyle{abbrv}
\bibliography{Bibliography}

\end{document}